\documentclass[11pt]{article} 

\usepackage[ansinew]{inputenc}
\usepackage{amsmath, amssymb, graphics, amsthm}
\usepackage[pdftex]{graphicx, color}
\usepackage{epsfig}
\usepackage{color}
\usepackage{undertilde}
\usepackage{fancyhdr}

\newtheorem*{Theorem}{Theorem}

\newcommand{\m}{\mbox{}}

\newcommand{\be}{\begin{equation}}
\newcommand{\ee}{\end{equation}}
\newcommand{\ba}{\begin{eqnarray}}
\newcommand{\ea}{\end{eqnarray}}

\title{{\sf New Variables for Classical and Quantum Gravity}\\
{\sf in all Dimensions I.  Hamiltonian Analysis}} 
\author{
{\sf N. Bodendorfer}$^1$\thanks{{\sf 
norbert.bodendorfer@gravity.fau.de}},
{\sf T. Thiemann}$^{1,2}$\thanks{{\sf 
thomas.thiemann@gravity.fau.de,
tthiemann@perimeterinstitute.ca}},
{\sf A. Thurn}$^1$\thanks{{\sf 
andreas.thurn@gravity.fau.de}}\\
\\
{\sf $^1$ Inst. for Theoretical Physics III, FAU Erlangen -- N\"urnberg,}\\
{\sf Staudtstr. 7, 91058 Erlangen, Germany}\\
\\
{\sf and}\\
\\
{\sf $^2$ Perimeter Institute for Theoretical Physics,}\\ 
{\sf 31 Caroline Street N, Waterloo, ON N2L 2Y5, Canada}
}
\date{{\small\sf \today}}

\makeatletter
\@addtoreset{equation}{section}
\makeatother

\begin{document} 

\maketitle

{\sf
\begin{abstract}

Loop Quantum Gravity heavily relies on a connection formulation of General Relativity such 
that 1. the connection Poisson commutes with itself and 2. the corresponding gauge group
is compact. This can be achieved starting from the Palatini or 
Holst action when imposing the time gauge. Unfortunately, this method is restricted to
$D+1=4$ spacetime dimensions. However, interesting String theories and Supergravity theories
require higher dimensions and it would therefore be desirable to have higher dimensional
Supergravity loop quantisations at one's disposal in order to compare these approaches.

In this series of papers we take first steps towards this goal. The present first paper develops 
a classical canonical platform for a higher dimensional connection formulation of the 
purely gravitational sector. The new ingredient is 
a different extension of the ADM phase space than the one used in LQG which 
does not require the time gauge and which generalises to any dimension $D > 1$. 
The result is a Yang -- Mills theory phase space subject to Gau{\ss}, spatial diffeomorphism
and Hamiltonian constraint as well as one additional constraint, called the simplicity constraint.
The structure group can be chosen to be SO$(1,D)$ or SO$(D+1)$ and the latter choice is 
preferred for purposes of quantisation.
\end{abstract}
}

\newpage

\tableofcontents

\vspace{20mm}

\section{Introduction}
\label{s1}

The quantisation of General Relativity remains one the most important open problems of contemporary physics. Early attempts to quantise the Hamiltonian formulation given by Arnowitt, Deser and Misner \cite{ArnowittTheDynamicsOf} have failed due to non-renormalisability \cite{GoroffQuantumGravityAt, GoroffTheUltravioletBehavior} among other problems. Supergravity in various dimensions entered the picture as a way to resolve these problems, however, not all could be addressed \cite{DeserTwoOutcomesFor, DeserInfinitiesInQuantum, DeserNonrenormalizabilityOfLast}. Meanwhile, Superstring theory \cite{GreenBook1, GreenBook2} and M-Theory \cite{PolchinskiBook1, PolchinskiBook2} have been proposed as theories of quantum gravity. They constrain the spacetime dimension to $D+1=10$ (Superstring theory) or $D+1=11$ (M-Theory) and symmetry arguments suggest that the respective Supergravities are their low energy limits \cite{GreenBook1, GreenBook2}. It is therefore interesting to (loop)-quantise these Supergravities
as a new approach to quantising the low-energy limit of Superstring theory or M-theory. 

However, the programme of loop quantisation (see e.g. \cite{ThiemannModernCanonicalQuantum} and references therein) requires the theory to be formulated in terms of a gauge theory. The reason for that is that 
only for theories based on connections and conjugate momenta background independent
Hilbert space representations have been found so far, which also support the constraints 
of the theory as densely defined and closable operators. Of course, a connection formulation 
is also forced on us if we want to treat fermionic matter as well.  A connection formulation for gravity in $D+1>4$ that can be satisfactorily quantised, even in the vacuum case, has not been given so far. For the case $D+1=4$, it was only in 1986 that Ashtekar discovered his new variables for General Relativity 
\cite{AshtekarNewVariablesFor}. The most important property of these variables is that the connection $A$ 
used has a canonically conjugate momentum $E$ such that $(A,E)$ have standard canonical
brackets, in particular the connection Poisson commutes with itself. This is not trivial. Indeed, the naive connection that one would expect from the first order Palatini formulation does not have this 
crucial property, because the canonical formulation of Palatini gravity suffers from second class
constraints and the Palatini connection then has non trivial corresponding Dirac brackets.
This prohibited so far to find Hilbert space representations, in particular those of LQG type 
in which the connection is represented as a multiplication operator, for these Palatini connection
formulations. The Ashtekar connection
does not suffer from this problem because it is the self-dual part of the 
Palatini connection (or spin connection in the absence of torsion terms). Unfortunately, for the 
only physically interesting case of Lorentzian signature this Ashtekar connection takes values
in the non compact SL$(2,C)$ rather than a compact group and again it is very difficult to find 
Hilbert space representations of gauge theories with non compact structure groups. 
 
As observed by Barbero \cite{BarberoRealAshtekarVariables}, a possible strategy to deal with this non compactness problem is to use the time gauge and to 
gauge fix the boost part of SO$(1,3)$. The resulting connection, which can be seen as the 
self dual part of the spin connection for Euclidean signature,  is then an SU$(2)$ connection.
The price to pay is that the Hamiltonian constraint for Lorentzian 
signature in terms of these variables is more complicated than in terms of the complex valued 
ones. However, this does not pose any problems in its quantisation \cite{ThiemannQSD1}.
Using these variables (which also allow a one parameter freedom related to the Immirzi
parameter \cite{ImmirziRealAndComplex}) a rigorous quantisation of General Relativity with a unique Hilbert space representation could be derived \cite{AshtekarRepresentationsOfThe, AshtekarRepresentationTheoryOf, LewandowskiUniquenessOfDiffeomorphism, FleischhackRepresentationsOfThe}. 
A different way to arrive at the same formulation is to start from the geometrodynamics phase 
space coordinatised by the ADM variables (three metric and extrinsic curvature) and to expand it by introducing (densitised) triads $E$ and conjugate momenta $K$ (basically the extrinsic curvature 
contracted with the triad). The connection is then the triad spin connection $\Gamma$ plus this conjugate
momentum, that is, $A=\Gamma+\gamma K$ where $\gamma$ is the real valued Immirzi parameter. The first miracle that happens in $3$ spatial dimensions 
is that this is at all possible: While $K$ transforms in the defining representation of SO$(3)$,
$\Gamma$ transforms in the adjoint representation of SO$(3)$. But for the case of SO$(3)$, these are
isomorphic and enable to define the object $A$. The second miracle that happens in $3$ spatial
dimensions is that this connection is Poisson self commuting which is entirely non trivial. 
Notice that in three spatial dimensions, the expansion of the phase space alters the number of configuration degrees of freedom from 
six per spatial point (described by the three metric tensor) to nine (described by the co-triad).
To get back to the original ADM phase space, one therefore has to add three constraints
and these turn out to comprise precisely an SU$(2)$ Gau{\ss} constraints just as in Yang Mills theory.

It is clear that this strategy can work only in $D=3$ spatial dimensions: A metric 
in $D$ spatial dimensions has $D(D+1)/2$ configuration degrees of freedom per spatial point
while a $D$-bein has $D^2$. We therefore need $D^2-D(D+1)/2=D(D-1)/2$ constraints which is 
precisely the dimensionality of SO$(D)$. However, an SO$(D)$ connection has $D^2(D-1)/2$ degrees
of freedom. Requiring that connection and triad have equal amount of degrees of freedom
leads to the unique solution $D=3$. Thus in higher dimensions we need a generalisation of 
the procedure that works in $D=3$.  Attempts to construct a higher dimensional connection formulation have been undertaken, but few results are available \cite{NietoTowardsAnAshtekar8, NietoTowardsAnAshtekar12, NietoOrientedMatroidTheory, MeloschNewCanonicalVariables}. Han et al. \cite{HanHamiltonianAnalysisOf} have shown that the higher dimensional Palatini action leads to geometrodynamics when the time gauge is imposed before the canonical analysis. 

In this paper, we will derive a connection formulation for higher dimensional General Relativity by 
using a different extension of the ADM phase space than the one employed in 
\cite{AshtekarNewVariablesFor, AshtekarNewHamiltonianFormulation} and which generalises to arbitrary spacetime dimension 
$D+1$ for $D>1$.
It is based in part on Peldan's seminal work \cite{PeldanActionsForGravity} on the possibility of using 
higher dimensional gauge groups for gravity as well as on his concept of a hybrid spin connection
which naturally appears in the connection formulation of $2+1$ gravity \cite{CarlipQuantumGravityIn}. 
More precisely, the idea is the following:\\
If one starts from the Palatini formulation in $D+1$ spacetime dimensions, then the natural 
gauge group to consider is SO$(1,D)$ or SO$(D+1)$ respectively for Lorentzian or Euclidean
gravity respectively. Both groups have dimension $D(D+1)/2$. This motivates to look
for a connection formulation of the Hamiltonian framework with a connection $A_{aIJ},\;
a=1,..,D;\;I,J=0,..,D$. Such a connection has $D^2(D+1)/2$ degrees of freedom. The corresponding Gau{\ss} constraint removes $D(D+1)/2$ degrees of freedom, leaving us 
with $(D-1)D(D+1)/2$ degrees of freedom. However, a metric in $D$ spatial dimensions 
has only $D(D+1)/2$ degrees of freedom, which means that we need 
$D^2(D-1)/2-D$ additional constraints which together with the ADM constraints and 
the Gau{\ss} constraint form a first class system. To discover this constraint, we need an 
object that transforms in the defining representation of the gauge group. It is given by 
the ``square root'' of the spatial metric $q_{ab}=\eta_{IJ} e_a^I e_b^J$ where $\eta$ 
has Lorentzian or Euclidean signature respectively. Since the $D$ internal vectors $e_a^I$ 
are linearly independent, we can complete them to a uniquely defined $(D+1)$-bein
by the unit vector $e_0^I$ where $\eta_{IJ}  e_a^I e_0^J=0$. Now the momentum 
$\pi^{aIJ}$ conjugate to $A_{aIJ}$ is supposed to be entirely determined by $e_a^I$,
that is, $\pi^{aIJ}\propto \sqrt{\det(q)} q^{ab} e_0^{[I} e_b^{J]} $. In other words,
$\pi$ is ``simple'' and we call these constraints therefore simplicity constraints. 
Since $e_a^I$ has $D(D+1)$
degrees of freedom while $\pi^{aIJ}$ has $D^2(D+1)/2$ these present precisely the required
$D^2(D-1)/2 -D$ constraints. Furthermore, from $e_a^I$ one can construct the hybrid 
spin connection $\Gamma_{aIJ}$ which annihilates $e_a^I$ and the idea, as for Ashtekar's 
variables, is that $A-\Gamma$ is related to the extrinsic curvature. In order to show that
the symplectic reduction of this extension of the ADM phase is given by the ADM phase space, similar to what happens in case of Ashtekar's variables, we need that $\Gamma$ is integrable at least 
modulo the simplicity constraints which we show to be the case.\\
It should be stressed that even in $D+1=4$ this extension of the ADM phase space is different from the one employed in LQG: In LQG the Ashtekar-Barbero connection is given by 
$A^{{\rm LQG}}_{ajk}-\Gamma_{ajk} \propto \epsilon_{jkl} K_a^l$, $i,j,k = 1,...,D$, while in our case in the time 
gauge $e_0^I=\delta^I_0$ we have 
 $A^{{\rm NEW}}_{ajk}-\Gamma_{ajk}$ is pure gauge. Here $\Gamma_{ajk}$ is the spin connection of 
 the corresponding triad. Thus, in the new formulation the information about the extrinsic 
 curvature sits in the $A_{a0j}$ component which is absent in the LQG formulation.
 We also emphasise that it is possible to have gauge group SO$(D+1)$ even for the Lorentzian 
 ADM phase space. While a Lagrangian formulation is only available when spacetime
 and internal signature match as we will see in a companion paper \cite{BTTII}, this opens 
 the possibility to quantise gravity in $D+1$ spacetime dimensions using LQG methods 
 albeit with structure group SO$(D+1)$ and additional (simplicity) constraints \cite{BTTIII, BTTV}.\\  
\\
This paper is is organised as follows:\\ 
\\
In section \ref{s2}, we will define the required kinematical structure of a $(D+1)$-dimensional 
connection formulation of General Relativity. We will study in detail the properties of the 
simplicity constraint and the hybrid spin connection.

In section \ref{s3}, we will postulate an extension of the ADM phase space in terms of 
a connection and its conjugate momentum subject to the corresponding Gau{\ss} constraint and the simplicity 
constraint discussed before. We will then prove that the symplectic reduction of this extension
with respect to both constraints recovers the ADM phase space. There is a one parameter 
freedom in this extension, similar to but different from the Immirzi parameter of standard 
LQG \cite{ImmirziRealAndComplex}.

In section \ref{s4}, we express the spatial diffeomorphism constraint and the Hamiltonian 
constraint in terms of the new variables and prove that the full set of four types of 
constraints, namely Gau{\ss}, simplicity, spatial diffeomorphism and Hamiltonian constraints,
is of first class. This can be done for either choice of SO$(1,D)$ or 
SO$(D+1)$ independently of the spacetime signature. Similar to the situation with standard 
LQG, the Hamiltonian simplifies when spacetime signature and internal signature match
if one chooses unit Immirzi like parameter. There is an additional correction term present
which accounts for the removal of the pure gauge degrees of freedom affected by the gauge 
transformations generated by the simplicity constraint.

In section \ref{s5}, we conclude and discuss future research directions partly already
addressed in our companion papers \cite{BTTII,BTTIII,BTTV,BTTIV,BTTVI, BTTVII}.\\
\\
The further organisation of this series is as follows:\\
\\
In paper \cite{BTTII}, we supplement the present paper by a Lagrangian, namely Palatini, formulation 
in the case that internal and spacetime signature match. As is well known, the canonical 
treatment of Palatini gravity, while leading to a connection formulation, is plagued
by second class constraints and a non trivial Dirac bracket prohibiting a connection representation
in the quantum theory \cite{AlexandrovSU(2)LoopQuantum}. This is in apparent contradiction to the first class structure found in the present paper. The link between the two approaches is through the machinery of gauge unfixing \cite{MitraGaugeInvariantReformulationAnomalous, MitraGaugeInvariantReformulation, AnishettyGaugeInvarianceIn, VytheeswaranGaugeUnfixingIn}, which transforms a second class system into an equivalent first class 
system subject to a modification of the Hamiltonian which in our case is precisely the 
additional correction term in the Hamiltonian constraint found in this paper.

In paper \cite{BTTIII} we quantise the constraints found in this paper using the standard machinery 
developed for the $D=3$ case. In paper \cite{BTTIV}, we consider coupling to fermionic matter and its 
quantisation where we have to 
solve the problem of how to switch from Lorentzian to Euclidean signature Clifford algebras.  In paper \cite{BTTV}, we quantise the simplicity constraint and show that in $D+1=4$
the resulting Hilbert space and the representation of Gau{\ss} and simplicity invariant observables 
coincides with the standard LQG representation.
In paper \cite{BTTVI}, we consider the classical canonical formulation of higher dimensional Supergravity theories, in particular the Rarita-Schwinger fields in terms of new canonical 
variables and formulate the corresponding
quantum theory. Finally, we treat the p-form sector of Supergravity theories in paper \cite{BTTVII}.

\section{Kinematical Structure of $(D+1)$-dimensional Canonical Gravity}
\label{s2}

This section is subdivided into three parts. In the first part we show that simple dimensional
counting and  
natural considerations lead to a unique candidate connection formulation that works
in any spacetime dimension $D+1$ and has underlying structure group SO$(D+1)$ or 
SO$(1,D)$ respectively. We also identify the simplicity constraints additional to the Gau{\ss} 
constraint that such a formulation requires and show that while there is no $D$-bein and no spin 
connection in such a formulation, there is a generalised $D$-bein and a hybrid connection. 
The latter is required in order to express the ADM variables in terms of the connection
and its conjugate momentum.
In the second part we formulate an equivalent expression for the simplicity constraint and discuss
its properties and some subtleties. Finally, in the third part we prove a key property of the hybrid connection, namely its 
integrability modulo simplicity constraints. This will be key to proving in the next section that the symplectic 
reduction of the extended phase space by Gau{\ss} and simplicity constraints recovers the 
ADM phase space. 

\subsection{Preliminaries}
\label{s2.1}

As is well known (see e.g. \cite{ThiemannModernCanonicalQuantum} and references therein), the ADM Hamiltonian 
formulation of vacuum $D+1$ General Relativity
is based on a phase space coordinatised by a canonical pair $(q_{ab},P^{ab})$ with non trivial
Poisson brackets (we set the gravitational constant to unity for convenience)
\be \label{2.1}
\{q_{ab}(x),P^{cd}(y)\}=\delta_{(a}^c\;\delta_{b)}^d\; \delta^{(D)}(x-y) \text{,}
\ee 
where $a,b,c,..\in \{1,..,D\}$ and $x,y,..$ are coordinates on a $D$-dimensional manifold 
$\sigma$. The images of $\sigma$ under one parameter families of embeddings of $\sigma$
into a $(D+1)$-dimensional manifold $M$ constitute a foliation of $M$. Here $q_{ab}$ is a metric 
on $\sigma$ of Euclidean signature. The phase space defined 
by  (\ref{2.1}) is subject to spatial diffeomorphism constraints
\be \label{2.2}
\mathcal{H}_a=-2 q_{ac}\; D_b P^{bc}
\ee
and Hamiltonian constraint
\be \label{2.3}
\mathcal{H}=-\frac{s}{\sqrt{\det(q)}}[q_{ac} q_{bd}-\frac{1}{D-1} q_{ab} q_{cd}]P^{ab} P^{cd}-
\sqrt{\det(q)} R^{(D)} \text{,}
\ee
where $R^{(D)}$ is the Ricci scalar of $q_{ab}$ and $D_a$ denotes the torsion free
covariant derivative annihilating $q_{ab}$. Here $s$ is the signature of the spacetime geometry. Expression (\ref{2.3}) is problematic for $D=1$ 
and in what follows we restrict to $D>1$.

Similar to the formulation of standard LQG in $D+1=4$, we would like to arrive at a connection formulation of this system which then can be quantised using standard LQG techniques.
This requires the corresponding structure group to be compact. Let us recall and sketch how this 
is done for $D=3$, see \cite{ThiemannModernCanonicalQuantum} for all the details: \\
Following Peldan \cite{PeldanActionsForGravity}, the idea is to extend the ADM phase space by additional degrees of freedom and then to impose 
additional first class constraints in such a way that the symplectic reduction of the extended system with respect to these constraints coincides with the original ADM phase space. 
In practical 
terms, this means that one considers a connection $A_a^\alpha$, i.e. a Lie algebra valued 
one form with a Lie algebra of dimension $N$ and a conjugate momentum $\pi^a_\alpha$ 
which is a Lie algebra valued vector density. Here $\alpha, \beta, .. =1,..,N$. Such a Yang-Mills 
phase space is subject to a Gau{\ss} constraint 
\be \label{2.4}
G_\alpha={\cal D}_a \pi^a_\alpha=\partial_a \pi^a_\alpha+f_{\alpha\beta}\;^\gamma\; A_a^\beta\; \pi^a_\gamma \text{,}
\ee
where $f_{\alpha\beta}\;^\gamma$ denote the structure constants of the corresponding
gauge group. The requirement is then that there is a  reduction $(A,\pi)\mapsto 
q_{ab}:=q_{ab}[A,\pi],\;P^{ab}:=P^{ab}[A,\pi]$ such that the Poisson brackets of the ADM 
phase space are reproduced modulo the Gau{\ss} constraint and possible additional first class  
constraints that maybe necessary in order that the correct dimensionality of the reduced phase space is achieved. 

The question is of course which group should be chosen depending on $D$ and 
how to express $q_{ab}, P^{ab}$ in terms of $A_a^\alpha,\;\pi^a_\alpha$. Furthermore, 
one may ask whether the Gau{\ss} constraint is sufficient in order to reduce to the correct number
of degrees of freedom or whether there should be additional constraints. Consider first the 
case that the Gau{\ss} constraint is sufficient. Then the extended phase space has $DN$ configuration
degrees of freedom of which the Gau{\ss} constraint removes $N$. This has to agree with 
the dimension of the ADM configuration degrees of freedom which in $D$ spatial dimensions 
is $D(D+1)/2$. It follows $N(D-1) = D(D+1)/2$. Next we need to relate $(A_a^\alpha,
\pi^a_\alpha)$ to $(q_{ab},P^{ab})$. There may be many possibilities for doing so but here 
we will follow a strategy that is similar to the strategy of standard LQG. We
consider some representation $\rho$ of the corresponding Lie group $G$ of dimension 
$M\ge D$ and introduce 
generalised $D$-beins $e_a^I,\; I,J,K,...=1,..,M$ taking values in this representation with 
$q_{ab}=e^I_a \eta_{IJ} e^J_b$. The requirement $M \ge D$ is needed in order that 
$q_{ab}$ can be chosen to be non degenerate and we furthermore require that it is 
positive definite. Here $\eta$ is a 
$G$-invariant tensor, i.e. $\rho(g)^I_K \eta_{IJ} \rho(g)^{J}_{L}=\eta_{KL}$. The existence of such 
a tensor already severely restricts the possible choices of $G$ and typically $G$ is simply
defined in this way whence $\rho$ will typically be  the defining representation of $G$. We extend the 
covariant derivative $D_a$ to $\rho$ valued objects by asking that $D_a$ annihilates the 
co-$D$-bein
\be \label{2.5}
D_a e_b^I =\partial_a e_b^I-\Gamma^c_{ab} e_c^I+\Gamma_a^\alpha\; [X^\rho_\alpha]^I \m_J
e_b^J=0 \text{,}
\ee
with the Levi-Civita connection $\Gamma^c_{ab}$. This equation 
defines the hybrid (or generalised) spin connection $\Gamma_a^\alpha$. Here the 
$X^\rho_\alpha$ denote the generators of the Lie algebra of $G$ in the representation $\rho$. 

The idea is now that $\tilde{K}_a\m^b:=[A_a^\alpha-\Gamma_a^\alpha]\pi^b_\alpha$ is the 
expression for the ADM extrinsic curvature $\sqrt{\det(q)}K_a\m^b,\;\;P_a\m^b=-s \sqrt{\det(q)}[K_a\m^b
-\delta_a^b K_c\m^c]$, in terms of the new variables. However, there are several caveats.
First of all, it is not clear that (\ref{2.5}) has a non-trivial solution: These are $D^2 M$ equations 
for $DN$ coefficients $\Gamma_a^\alpha$ and thus the system (\ref{2.5}) could be overdetermined. Secondly, even if a solution exists, 
$\Gamma_a^\alpha$ will be a function of $e_a^I$ while we need to express it in terms of the  momentum $\pi^a_\alpha$ conjugate to $A_a^\alpha$. If there is no other constraint than 
the Gau{\ss} constraint, then $\pi^a_\alpha$ itself must be already determined in terms 
of $e_a^I$ which implies that $M=N$: The representation $\rho$ has the same dimension
as the adjoint representation of the Lie group. If one scans the classical Lie groups, then 
the only case where the defining representation and the adjoint representation have the 
same dimension (and are in fact isomorphic) is SO$(3)$ or SO$(1,2)$ respectively, whence $N=3$. In this case, the equation $N(D-1) = D(D+1)/2$ has the solutions $D=2$ and $D=3$ which can be shown to be the only solutions to this equation on the positive integers.

In order to go beyond $D=3$, we therefore need more constraints. We consider now the case 
of the choice $G=\text{SO}(M+1)$ or $G=\text{SO}(1,M)$ which is motivated by the fact that these Lie groups underly 
the Palatini formulation of General Relativity in $M+1$ spacetime dimensions. 
Following Peldan's programme, other choices may be leading, conceivably, to canonical 
formulations of GUT theories. We will leave the investigation of such possibilities for future research.  For this choice we obtain 
$N=M(M+1)/2$ and thus (\ref{2.5}) presents $D^2(M+1)$ equations for $DM(M+1)/2$ 
coefficients.  Explicitly
\be \label{2.6}
\partial_a e_b^I-\Gamma^c_{ab} e_c^I+\Gamma_a^{IJ}\;e_{bJ}=0 \text{,}
\ee
where all internal indices are moved with $\eta$. Since $\Gamma_{a(IJ)}=0$ we obtain the 
consistency condition 
\be \label{2.7}
e_{(cI}\partial_a e_{b)}^I-\Gamma_{(c|a|b)} =0 \text{,}
\ee
where $q_{ab}=e^I_a e_{bI}$ was used. It is not difficult to see that (\ref{2.7}) is in fact 
identically satisfied. Therefore the $D^2(M+1)$ equations (\ref{2.6}) are not all independent,
there are $D^2(D+1)/2$ identities (\ref{2.7}) among them, reducing the number of independent equations to $D^2[M+1-\frac{1}{2}(D+1)]$ for $DM(M+1)/2$ coefficients $\Gamma_{aIJ}$. 
Equating the number of independent equations to the number of equations yields a quadratic 
equation for $M$ with the two possible roots $M=D$ and $M=D-1$. In the second case 
$e_a^I$ is an ordinary $D$-bein and $\Gamma_{aIJ}$ its ordinary spin connection. In the former
case we obtain the hybrid spin connection mentioned before.

Let us discuss the cases SO$(D)$ and SO$(D+1)$ separately (the discussion is analogous for 
SO$(1,D-1)$ and SO$(1,D)$ except that SO$(1,D-1)$ does not allow for a positive definite 
$D$ metric and therefore must be excluded anyway). In the case of SO$(D)$ we have $D^2(D-1)/2$
configuration degrees of freedom and $D(D-1)/2$ Gau{\ss} constraints. In order to match the 
number of ADM degrees of freedom, we therefore need  
$S=D^2(D-1)/2-D(D-1)/2-D(D+1)/2=D^2(D-3)/2$ additional constraints. These must be imposed 
on the momentum $\pi^{aIJ}$ conjugate to $A_{aIJ}$ and require that $\pi^{aIJ}$ is already 
determined by $e_a^I$. Now $e_a^I$ has $D^2$ degrees of freedom while $\pi^{aIJ}$ has 
$D^2(D-1)/2$ so that exactly $S$ degrees of freedom are superfluous. However, there is no
way to to build an object $\pi^{aIJ}$ with $\pi^{a(IJ)} = 0$ from $e_a^I$: In order to match the density
weight we can consider $E^{aI}=\sqrt{\det(q)} q^{ab} e_b^I$, but we cannot algebraically build another object $v^I$ from $e_a^I$ without tensor index in order to define $\pi^{aIJ}=v^{[I} E^{a|J]} $.
The only solution is that there are no superfluous degrees of freedom, which leads back
to $D=3$. Now consider SO$(D+1)$. In this case we have $D^2(D+1)/2$ configuration degrees
of freedom and $D(D+1)/2$ Gau{\ss} constraints requiring 
$S=D^2(D+1)/2-D(D+1)/2-D(D+1)/2=D^2(D-1)/2-D$ additional constraints. The number of superfluous
degrees of freedom in $\pi^{aIJ}$ as compared to $e_a^I$ is now also 
precisely $S=D^2(D+1)/2 -D(D+1)$. In contrast to the previous case, however, now it is possible 
to construct an object without tensor indices: If we assume that the $D$ internal vectors 
$e_a^I,\;a=1,..,D$ are linearly independent then we construct the common normal
\be \label{2.8}
n_I:=\frac{1}{D!} \frac{1}{\sqrt{\det(q)}}\epsilon^{a_1..a_D} \epsilon_{IJ_1..J_D}
e_{a_1}^{J_1}\;..\;e_{a_D}^{J_D} \text{,}
\ee 
which satisfies $e_a^I n_I=0,\; n_I n^I=\zeta$ where $\zeta=1$ for SO$(D+1)$ and $\zeta=-1$ for 
SO$(1,D)$. Notice that $n_I$ is uniquely (up to sign) determined by $e_a^I$. We may now require 
that 
\be \label{2.9}
\pi^{aIJ}=2\sqrt{\det(q)} q^{ab} n^{[I} e_b^{J]} =:2 n^{[I} E^{a|J]}\text{.}
 \ee
 These are  the searched for constraints on $\pi^{aIJ}$ and constitutes our candidate 
 connection formulation for General Relativity in arbitrary spacetime dimensions $D+1\ge 3$.
 Since they require $\pi$ to come from 
 a generalised $D$-bein, we call them {\it simplicity constraints}.
 Notice that $D^2(D-1)/2-D=0$ for $D=2$. Indeed, $2+1$ gravity is naturally defined as 
 an SO$(1,2)$ or SO$(3)$ gauge theory.

 \subsection{Properties of the Simplicity Constraints}
\label{s2.2}

 The form of the constraint (\ref{2.9}) is not yet satisfactory because the constraint should be 
 formulated purely in terms of $\pi^{aIJ}$. The same requirement applies to the hybrid connection
 to which we will turn in the next subsection. \\
 \\
 Given $\pi^{aIJ}$ and any unit vector $n_I$ we may define $E^{aI}[\pi,n]:= - \zeta \pi^{aIJ} n_J$. This 
 object then automatically satisfies $E^{aI} n_I=0$. Furthermore we may define the 
 transversal projector 
 \be \label{2.10}
 \bar{\eta}^I_J[n]:=\delta^I_J-\zeta n^I n_J\;\;\Rightarrow\;\; \bar{\eta}^I_J \; n^J=0
 \ee
 and define 
 \be \label{2.11}
 \bar{\pi}^{aIJ}:=\bar{\eta}^I_K[n]\;\bar{\eta}^J_L[n]\; \pi^{aKL} \text{.}
 \ee
 In what follows, all tensors with purely transversal components will carry an overbar. We obtain 
 the decomposition
 \be \label{2.12}
 \pi^{aIJ}=\bar{\pi}^{aIJ}+2 n^{[I} E^{a|J]} \text{.}
 \ee
 It appears that the simplicity constraint now is equivalent to $\bar{\pi}^{aIJ}=0$. However, there 
 are two subtleties: First, at this point $n^I$ is an extra structure next to $\pi^{aIJ}$ which is
 required to define (\ref{2.11}). Therefore the decomposition (\ref{2.12}) is not intrinsic
 and $n^I$ appears as an extra degree of freedom. It is therefore necessary to give an 
 intrinsic definition of $n^I$. Next, suppose that we have achieved to do so, then 
 $\bar{\pi}^{aIJ}$ constitute $D^2(D-1)/2$ degrees of freedom rather than the required
 $D^2(D-1)/2-D$ while due to $E^a_I n^I=0$ the $E^a_I$ constitute only $D^2$ degrees 
 of freedom rather than $D(D+1)$. 
 
 To remove these subtleties, it is cleaner to adopt the following point of view:\\ 
 We consider $D+1$ vector densities $E^a_I$ to begin with such that the corresponding \mbox{$D(D+1)$-matrix} has maximal rank. From these we can construct the 
 densitised inverse metric 
 \be \label{2.13}
 Q^{ab}:=E^a_I E^b_J \eta^{IJ}\text{,}
 \ee
 which we require to have Euclidean signature
 as well as their common normal 
 \be \label{2.14}
 n_I[E]:=\frac{1}{D!\;\sqrt{\det(Q)}}\; \epsilon_{a_1..a_D}\;\epsilon_{IJ_1..J_D} 
 E^{a_1 J_1} .. E^{a_D J_D} \text{,}
 \ee
 which is now considered as a function of $E$. Notice that $n_I n^I=\zeta$. Therefore also 
 $\bar{\eta}^I_J=\bar{\eta}^I_J[E]$ is a function of $E$. We can again apply the 
 decomposition (\ref{2.12}) and now have cleanly deposited the searched for 
 degrees of freedom into $E^a_I$.  However, while $n^I$ is now intrinsically defined via 
 $E^a_I$, the constraints $\bar{\pi}^{aIJ}=0$ are still $D$ to many. We should remove $D$ additional
 degrees of freedom from $\bar{\pi}^{aIJ}$. To do so we impose a tracefree condition. 
 Consider the object
 \be \label{2.15}
 E_a^I:=Q_{ab} E^{bI},\;\;Q_{ac} Q^{cb}:=\delta_a^b \text{.}
 \ee
 It follows easily from the definitions that
 \be \label{2.16}
 E_a^I E^b_I=\delta_a^b,\;\;E_a^I E^a_J=\bar{\eta}^I_J \text{.}
 \ee
 Consider the tracefree, transverse projector 
 \be \label{2.17}
 P^{aIJ}_{bKL}[E]:=\delta^a_b\bar{\eta}^I_{[K} \bar{\eta}^J_{L]}-\frac{2}{D-1} E^{a[I} \; E_{b[K}
 \;\bar{\eta}^{J]}_{L]} \text{.}
 \ee
 Then for any tensor $\pi^{aIJ}$ we have with $\bar{\pi}^{aIJ}_T=
 P^{aIJ}_{bKL} \pi^{aIJ}$ that 
 \be \label{2.18}
 \bar{\pi}^J:=E_{aI}\bar{\pi}^{aIJ}_T=0
 \ee
 and $\bar{\pi}^{aIJ}_T n_I=0$. Notice that $\bar{\pi}^{aIJ}_T$ has only $D^2(D-1)/2-D$ degrees
 of freedom independent of $E^a_I$.\\ 
 \\
 We therefore consider in what follows tensors $\pi^{aIJ}$ of the following form 
 \be \label{2.19}
 \pi^{aIJ}[E,\bar{S}_T]:=\bar{S}^{aIJ}_T+2 n^{[I}[E] \, E^{a|J]} \text{,}
 \ee
 where $\bar{S}_T$ and $E$ are considered as independent parameters for $\pi$. Notice that 
 $\bar{S}_T$ can be constructed as $P\cdot S$ from an arbitrary tensor $S^{aIJ}$.
 Such tensors can be intrinsically described as follows:\\
 Given $\pi$, there exists a normal $n_I[\pi]$ such that the following holds: Define 
 $E^a_I[\pi,n]= - \zeta \pi^{aIJ}n_J$ and $\bar{\pi}^{aIJ}[\pi,n]$ as above. Then automatically
 \be \label{2.20}
\bar{\pi}^J[\pi,n]:= \bar{\pi}^{aIJ}[\pi,n] Q_{ab}[\pi,n] E^b_I[\pi,n]=0 \text{.}
 \ee
This is a set of $D$ independent (since automatically $\bar{\pi}^I n_I=0$ no matter what $n^I$ is), non-linear equations  for the $D$ independent (due to the normalisation $n_I n^I=\zeta$) components of 
$n^I$.
In the appendix, we study this non trivial system of equations further and show that it can possibly
be solved by fixed point methods. At present we do not know 
whether at least tensors $\pi^{aIJ}$ subject to the condition that $\zeta \pi^{aIJ}\pi^b_{IJ}/2$ is 
positive definite always allow for such a solution $n^I$, however, we know that the number 
of possible solutions is always finite because we can transform (\ref{2.20}) into a system 
of polynomial equations. In what follows, we will assume that the solution $n^I[\pi]$ is 
in fact unique by suitably restricting the set of allowed tensors $\pi^{aIJ}$. This could 
imply that the set of such tensors no longer has the structure of a vector space which 
however does not pose any problems for what follows.\\
\\
On the other hand, we can prove the following for general $\pi^{aIJ}$:
\newpage
\begin{Theorem} \label{th2.1} ~\\
Let $D\ge 3$ and\footnote{For $D=2$ no simplicity constraints are needed since 
$D^2(D-1)/2-D=0$.}
\be \label{2.21}
S^{ab}_{\overline{M}}:=\frac{1}{4} \epsilon_{IJKL\overline{M}} \pi^{aIJ} \pi^{bKL} \text{,}
\ee
where $\overline{M}$ is any totally skew $(D-3)$-tuple of indices in $\{0,1,..,D\}$. Then
\be \label{2.22}
S^{ab}_{\overline{M}}=0\;\;\forall \;\; \overline{M}, \;a, \;b \;\;\;\Leftrightarrow\;\;P^{aIJ}_{bKL}[\pi,n]\;\pi^{bKL}=0
\ee
for any unit vector $n$ where $P^{aIJ}_{bKL}[\pi,n]:=[P^{aIJ}_{bKL}[E]]_{E=E[\pi,n]}$ and 
$E^{aI}[\pi,n]=- \zeta \pi^{aIJ} n_J$ and where $P[E]$ is defined in (\ref{2.17}).
Here we assume that $Q^{ab}[\pi,n]:=\pi^{aIK} \pi^{bJL} \eta_{IJ} n_K n_L$ is non degenerate
for any (timelike for $\zeta=-1$) vector $n_I$.
\end{Theorem}
This result implies that although $S^{ab}_{\overline{M}}$ are $D(D+1)/2\;{D+1 \choose 4}$
equations which exceeds $D^2(D-1)/2-D$ for $D>3$ only $D^2(D-1)/2$ of them are 
independent. The constraint $S^{ab}_{\overline{M}}=0$ does not fix $n^I$ and makes 
no statement about the trace part $\bar{\pi}^J[\pi,n]=\bar{\pi}^{aIJ}[\pi,n] E_{aI}[\pi,n]$.
Given that the theorem holds for any $n$ it is natural to fix $n$ such that the trace part 
vanishes simultaneously as otherwise we would have only that $\bar{\pi}^{aIJ}=2E^{a[I}\bar{\pi}^{J]}/(D-1)$
and not $\bar{\pi}^{aIJ}=0$ or $\pi^{aIJ}=2 n^{[I}E^{a|J]} $ on the constraint surface of the 
simplicity constraint.
\begin{proof} ~\\
Obviously 
\be \label{2.23}
S^{ab}_{\overline{M}}=0\;\; \Leftrightarrow\;\;\epsilon^{IJKL\overline{M}} S^{ab}_{\overline{M}}
= \frac{\zeta}{4}\: 4!\; (D-3)!\; \pi^{a[IJ}\;\pi^{bKL]}=0 \text{.}
\ee 
Given $\pi$, consider any unit vector $n$ and decompose as in (\ref{2.12}) 
\be \label{2.24}
\pi^{aIJ}=\bar{\pi}^{aIJ}[\pi,n]+2 n^{[I} E^{a|J]}[\pi,n] \text{.}
\ee
Inserting into (\ref{2.24}), we obtain 
\be \label{2.25a}
\pi^{a[IJ}\pi^{bKL]}=\bar{\pi}^{a[IJ}\bar{\pi}^{bKL]}+4 n^{[I} E^{(a|J} \bar{\pi}^{b)KL]}=0 \text{.}
\ee
Contracting with $n_I$ yields 
\be \label{2.25b}
E^{(a[J} \bar{\pi}^{b)KL]}=0 \text{.}
\ee
Contracting further with $E_{aJ}$ yields
\be \label{2.26}
(D-1)\;[\bar{\pi}^{bKL}-\frac{2}{D-1} E^{b[K} \bar{\pi}^{aJ|L]} E_{aJ}]
=(D-1) \;P^{bKL}_{aIJ}[\pi,n]\pi^{aIJ}=0 \text{.}
\ee
We conclude $\pi^{aIJ}=2 v^{[I} E^{a|J]},\; v^I=(n^I - \frac{1}{D-1} \bar{\pi}^{bJI}E_{bJ})$ and inserting 
back into (\ref{2.23}) we see that it is identically satisfied.
\end{proof}
The theorem therefore says that on the constraint surface $\pi^{aIJ}=2 v^{[I} E^{a|J]} $ for some 
vector $v$ which is not necessarily normalised and not necessarily normal to $E^{aI}$ but such 
that $E^{aI}, v^I$ constitute $D+1$ linearly independent internal vectors.  We can however draw,
for $\zeta=-1$,
some additional conclusion from the requirement that $Q^{ab}=\pi^{aIJ} \pi^b_{IJ}/(2\zeta)$ should 
have Euclidean signature. First of all, $v^I$ cannot be null since otherwise $Q^{ab}\propto
(E^a_I v^I) (E^b_J v^J)$ would be degenerate. If $v^I$ would be spacelike then 
consider $\tilde{E}^a_I=E^a_I-E^a_J v^J v_I/(v^K v_K)$. It follows $\pi^{aIJ}= 2v^{[I} \tilde{E}^{a|J]}$
and $Q^{ab}\propto \tilde{E}^{aI} \tilde{E}^b_I$. Since $v^I,\tilde{E}^a_I$ constitutes a $(D+1)$-bein
and $v^I$ is spacelike while $\eta$ is Lorentzian, also $Q^{ab}$ would need to be Lorentzian.
Hence $v^I$ must in fact be timelike for $\zeta=-1$.

We may therefore absorb for either signature the normalisation of $v$ into
$E^a_I$ and define $n_I:=v_I/\sqrt{\zeta v_J v^J}$ as well as 
$\tilde{E}^{aI}=\sqrt{\zeta v_K v^K}E^{aJ} \bar{\eta}^I_J$. Then 
$2 v^{[I} E^{a|J]} = 2 n^{[I} \tilde{E}^{a|J]} $ with $\tilde{E}^{aI} n_I=0, \; n^I n_I=\zeta$.\\ 
 \\
Therefore, the constraint surface defined via (\ref{2.21}) is the same as the one given by 
$\bar{\pi}^{aIJ}_T$  above, where we assumed 
that $\pi$ is of the form (\ref{2.19}) and constitutes the unique decomposition 
of $\pi^{aIJ}$ with no trace part.  In what follows, we will use the simplicity 
constraint in the form (\ref{2.21}). However, it will be convenient to have the presentation 
(\ref{2.19}) at one's disposal when we work off the constraint surface.\\
\\
Notice that the proof given above also in the case $D=3$ does not allow for a ``topological 
sector'' $\pi^{aIJ}=\epsilon^{IJKL} n^K E^{aL}$ or ``degenerate sector'' due to the 
non degeneracy assumption. This assumption is dropped in the alternative proof in 
\cite{BTTII} which is based on
\cite{FreidelBFDescriptionOf} which is why the topological sector does appear there.

\subsection{Integrability of the Hybrid Connection Modulo Simplicity Constraint}
\label{s2.3}

The hybrid connection is defined via (\ref{2.6}) on the constraint surface $S^{ab}_{\overline{M}}=0$.
We want to define an extension off the constraint surface such that the resulting expression 
is integrable, i.e. is the functional derivative $\Gamma_{aIJ}=\delta F/\delta \pi^{aIJ}$ of a generating functional  $F=F[\pi]$.  To that end, we need the explicit expression of 
$\Gamma_{aIJ}$ in terms of $e_a^I$.\\
\\
To begin with, we notice that $D_a n^I=0$. To see this we consider its $D+1$ independent 
components $n_I D_a n^I=\frac{1}{2}D_a (n^I n_I)=0$ and $e_b^I D_a n^I=-n^I D_a e_b^I=0$.
We decompose 
\be \label{2.27}
\Gamma_{aIJ}=\bar{\Gamma}_{aIJ}+2 n_{[I} \bar{\Gamma}_{a|J]},\; 
\bar{\Gamma}_{aI}=- \zeta \Gamma_{aIJ} n^J 
\ee
and further 
\be \label{2.28}
\bar{\Gamma}_{aIJ}=\bar{\Gamma}_{abc} e^b_I e^c_J,\;\;
\bar{\Gamma}_{aI}=\bar{\Gamma}_{ab} e^b_I \text{,}
\ee
with $e^b_I=q^{ab} e_{bI},\; q^{ac} q_{cb}=\delta^a_b,\;\;q_{ab}=e_a^I e_{bI}$. 
We find 
\be \label{2.29}
\bar{\Gamma}_{ab}=- \zeta n_I \partial_a e_b^I,\;\;
\bar{\Gamma}_{abc}=\Gamma_{bac}-e_{bI} \partial_a e_c^I\text{,}
\ee 
where $\Gamma_{bac}=q_{bd} \Gamma^d_{ac}$ is the Levi-Civita connection. Combining
these formulae, we obtain
\ba \label{2.30}
\Gamma_{aIJ}[E] &=&
-[\eta_{K[I}+\zeta n_K\; n_{[I}]e^b_{J]}\partial_a e_b^K+\Gamma^b_{ac} e_{b[I} e^c_{J]}
\nonumber\\
&=& \zeta n_{[I} \partial_a n_{J]}+e_{b[I} \partial_a e^b_{J]}+
\Gamma^b_{ac} e_{b[I} e^c_{J]} \text{,}
\ea 
where we used here and will also use frequently later 
$n_K \partial_a E^{bK}=-E^{bK} \partial_a n_K$, $n^K\partial_a n_K=0$ and 
$n_{[I} \bar{\eta}^K_{J]}=n_{[I} \eta^K_{J]}$. 

To write $\Gamma_{aIJ}$ in terms of $\pi^{aIJ}$, we notice the following weak identities
modulo the simplicity constraint, that is $\pi^{aIJ}\approx 2 n^{[I} E^{a|J]} $,
\ba \label{2.31}
\pi^{aIJ}\pi^b_{IJ} &\approx& 4 n^{[I} E^{a|J]}\; n_{[I} E^b_{J]}=2\zeta E^{aI} E^b_I=2\zeta Q^{ab} \text{,}
\nonumber\\
Q_{ab} \pi^{aKI}\pi^b_{KJ} & \approx & [ n^K E^{aI} - n^I E^{aK} ]\;[ n_K E_{aJ} - n_J E_{aK}] \nonumber \\
&=& D n^I n_J +\zeta \bar{\eta}^I_J=(D-1)n^I n_J +\zeta \eta^I_J \text{,}
\nonumber\\
E^{a[I} n^{J]} &=& - \zeta \pi^{a[I|L} \; n^{J]} n_L \text{,}
\nonumber\\
Q_{bd} \pi^{dK}\;_{[I}\; \pi^c_{K|J]} &\approx& 
[ n^K E_{b[I} - E_{b}^K n_{[I} ]\;[E^c_{J]} n_K - n_{J]} E^{c}_{K}]=\zeta\;E_{b[I} E^c_{J]}=\zeta \;e_{b[I} e^c_{J]} \text{,}
\nonumber\\
Q_{bc} \pi^{bK}\;_{[I}\partial_a \pi^c_{K|J]}
&\approx & [ n^K E_{c[I} - E_{c}^K n_{[I} ] \;\partial_a\; [E^c_{J]} n_K - n_{J]} E^c_K]
\nonumber\\
&=& -n^K E_{c[I} \; [n_{J]} (\partial_a E^c_K)-(\partial_a E^c_{J]}) n_K] \nonumber \\
& & + E_{c}^K n_{[I} \; [(\partial_a n_{J]}) E^c_K-E^c_{J]} (\partial_a n_K)] 
\nonumber\\
&=& (D-1)n_{[I} (\partial_a n_{J]}) +E_{c[I} n_{J]} E^c_K (\partial_a n^K)+\zeta E_{c[I} 
(\partial_a E^c_{J]})
\nonumber\\
&=&
(D-2) n_{[I} (\partial_a n_{J]}) +\zeta E_{c[I}  (\partial_a E^c_{J]})
\nonumber\\
&=&
(D-2) n_{[I} (\partial_a n_{J]}) +\zeta e_{c[I}  (\partial_a e^c_{J]}) \text{,}
\nonumber\\
\bar{\eta}_I^K \bar{\eta}_J^L
Q_{bd} \pi^{dM}\;_{[K}\; \partial_a \pi^c_{|M|L]} &\approx&
\zeta \bar{\eta}_{[I}^K \bar{\eta}_{J]}^L e_{b[K} \partial_a e^b_{L]}
\nonumber\\
&=& \zeta e_{b[I} \partial_a e^b_{J]}-n_{[I}\partial_a n_{J]} \text{.}
\ea
Consider the quantities
\be \label{2.33}
T_{aIJ}:=\pi_{bK[I} \partial_a \pi^{bK}\;_{J]},\;\;
T^c_{bIJ}:=\pi_{bK[I} \pi^{cK}\;_{J]} \text{,}
\ee 
where $\pi_{aIJ}=Q_{ab} \pi^b_{IJ}$. Then 
\be \label{2.34}
(D-1) n_{[I}\partial_a n_{J]}=T_{aIJ}-\bar{T}_{aIJ},\;\;
(D-1) \zeta e_{b[I} \partial_a e^b_{J]}=T_{aIJ}+(D-2)\bar{T}_{aIJ} \text{.}
\ee
Inserting (\ref{2.33}) and (\ref{2.34}) into (\ref{2.30}) then leads to the explicit expression
\be \label{2.35}
\Gamma_{aIJ}[\pi]=\frac{2\zeta}{D-1}T_{aIJ}+\frac{\zeta(D-3)}{D-1}\bar{T}_{aIJ}+\zeta \Gamma^b_{ac}
T^c_{bIJ} \text{.}
\ee
Together with $Q^{ab}=:\det(q) q^{ab}$ which expresses $\Gamma^b_{ac}$ in terms of
$Q^{ab}=\pi^{aIJ}\pi^b_{IJ}/(2\zeta)$,
 this determines $\Gamma_{aIJ}$ completely in terms of 
$\pi^{aIJ}$ if we simply replace the $\approx$ signs in (\ref{2.31}) by $=$ signs 
and take the left hand sides
as definitions for the right hand sides. \\
\\
It transpires that $\Gamma_{aIJ}$ is a rational, homogeneous function of $\pi$ and its first derivatives which vanishes at $\pi=0$. Therefore, if $\Gamma_{aIJ}[\pi]$
has a generating functional, then it is given by\footnote{If a one form $\Gamma_M$ 
is exact, i.e. has 
potential $U$ with $\Gamma_M=U_{,M}$ then $U(\pi)-U(\pi_0)=\int_{\gamma_{\pi_0,\pi}} \Gamma$
for any path $\gamma_{\pi_0,\pi}$ between $\pi_0$ and $\pi$. If $\Gamma$ is defined at $\pi_0=0$ to vanish then choosing the straight path $t\mapsto t\pi$ yields $U(\pi)=
const.+\int_0^1\; dt\pi^M \Gamma_M(t\pi)$.} 
\be \label{2.36}
F'[\pi]=\int\; d^Dx\; \pi^{aIJ}\; \Gamma_{aIJ}[\pi] \text{.}
\ee
Variation of $F'$ with respect to $\pi^{aIJ}$ yields 
\ba \label{2.37} 
\delta F' 
&=& \int\; d^Dx\left( \delta\pi^{aIJ}\; \Gamma_{aIJ}[\pi]+\pi^{aIJ} \delta\Gamma_{aIJ}[\pi]\right)
\nonumber\\
&=& \int\; d^Dx\left(\delta\pi^{aIJ}\; \Gamma_{aIJ}[\pi]
+\pi^{aIJ} [\delta\Gamma_{aIJ}[E]+\delta S'_{aIJ}]\right)
\nonumber\\
&=& \delta[\int\; d^Dx \pi^{aIJ} S'_{aIJ}]
+\int\; d^Dx\left( \delta\pi^{aIJ}\; \Gamma_{aIJ}[\pi]
+2 n^{[I} E^{a|J]} \;\delta\Gamma_{aIJ}[E]\right)
\nonumber\\
&&+\int\; d^Dx\; \left(S^{aIJ}\;\delta\Gamma_{aIJ}[E]-\delta\pi^{aIJ} \;S'_{aIJ}\right) \text{,}
\ea
where $S^{aIJ}:=\pi^{aIJ}-2 n^{[I} E^{a|J]} $ and 
$S'_{aIJ}:=\Gamma_{aIJ}[\pi]-\Gamma_{aIJ}[E]$ both vanish on the constraint surface 
of the simplicity constraint. We see that $F'$ itself cannot be a generating functional
but rather 
\be \label{2.38}
F=F'-\int\; d^Dx \, \pi^{aIJ} S'_{aIJ}\text{,}
\ee
i.e. $F'$ has to be corrected by a term that vanishes on the constraint surface 
of the simplicity constraint, however, its variation does not necessarily vanish on that 
constraint surface. 
It follows that $\delta F/\delta \pi^{aIJ}=\Gamma_{aIJ}+\tilde{S}_{aIJ}$ for some  $\tilde{S}_{aIJ}$
which vanishes on the constraint surface of the simplicity constraint provided that
\be \label{2.39}
\int d^Dx \, n^{[I} E^{a|J]} \delta\Gamma_{aIJ}[E]
=\int d^Dx\sqrt{\det(q)} n^{[I} e^{a|J]} \delta\Gamma_{aIJ}[E]=0 \text{.}
\ee
This is the key identity that one has to prove. It is the counterpart to the key identity that 
is responsible for the fact that the Ashtekar connection is Poisson commuting
in $D+1=4$. The reason for the correction $F'\to F$ is that $\Gamma_{aIJ}[\pi]$ 
is not strictly integrable but only modulo terms that vanish on the constraint surface of
the simplicity constraint.\\
\\
We proceed with the proof of (\ref{2.39}). 
It is easiest to use (\ref{2.27}) -- (\ref{2.29}). 
We have, using $n_K\delta n^K=0,\;n_K \delta e_b^K=-e_b^K\delta n_K$ and 
that $\bar{\Gamma}_{a(bc)}=0$,
\ba \label{2.40}
n^{[I} e^{a|J]} \delta(2 n_{[I} \bar{\Gamma}_{a|J]}) 
&=& 
2 n^{[I} e^{a|J]}  [n_I  (\delta\bar{\Gamma}_{aJ}) + \bar{\Gamma}_{aJ} \delta n_I)]
\nonumber\\
&=& 
\zeta e^{aI} (\delta\bar{\Gamma}_{aI}) = - e^{aI}\delta(n_J(\partial_a e_b^J) e^b_I) 
\nonumber\\
&=& e^{aI}\delta(e_b^J (\partial_a n_J) e^b_I)= e^{aI}\delta(\bar{\eta}^J_I \partial_a n_J)
\nonumber\\
&=& e^{aI} \nabla_a (\delta n_I) \text{,}
\nonumber\\
n^{[I} e^{a|J]} \delta\bar{\Gamma}_{aIJ}
&=& n^{I} e^{aJ} \delta(\bar{\Gamma}_{abc} e^b_I e^c_J)
\nonumber\\
 &=& n^{I} e^{aJ} \bar{\Gamma}_{abc} e^c_J  (\delta e^b_I)
 =q^{ac} \bar{\Gamma}_{acb} e^b_I  (\delta n^I)
\nonumber\\
 &=& - [e^a_J (\partial_a e_b^J)-\Gamma^a_{ab}] e^b_I (\delta n^I) \nonumber \\
 &=&- [e^a_J[ \partial_a(e^b_I  e_b^J)-e_b^J (\partial_a e^b_I)-\Gamma^a_{ab} e^b_I]\;(\delta n^I) 
 \nonumber\\
 &=&-  [e^a_J\partial_a(\bar{\eta}^J_I)-(\nabla_a e^a_I)](\delta n^I)
 = [\nabla_a e^a_I]\; [\delta n^I] \text{,}
 \ea
 where $\nabla_a$ is the torsion free covariant differential annihilating $q_{ab}$ (it acts 
 only on tensor indices, not on internal ones). We conclude
 \be \label{2.41a}
 \int\; d^Dx \; n^{[I} E^{a|J]}\; \delta\Gamma_{aIJ}[E]
 = \int\; d^Dx\; \sqrt{\det(q)} \nabla_a [e^a_I \delta n^I]= \int\; d^Dx \partial_a (E^a_I\delta n^I)
 =0
 \ee
 for suitable boundary conditions on $E^a_I$ and its variations\footnote{For instance
 one could impose that $n_I$ deviates from a constant by a function of rapid decrease 
 at spatial infinity.}.\\
 \\
 We therefore have established:
 \begin{Theorem} \label{th2.2} ~\\
 There exists a functional $F[\pi]$ such that for $\delta n^I$ vanishing sufficiently fast at 
 spatial infinity, we have 
 \be \label{2.41b}
 \delta F[\pi]/\delta\pi^{aIJ}(x)=\Gamma_{aIJ}[\pi;x)+S_{aIJ}[\pi;x) \text{,}
 \ee
 where $S_{aIJ}$ vanishes on the constraint surface of the simplicity constraint, depending
 at most on its first partial derivatives and 
 $\Gamma_{aIJ}[\pi]$ is the hybrid connection (\ref{2.35}).
 \end{Theorem}

 \section{New Variables and Equivalence with ADM Formulation}
 \label{s3}
 
 We consider an $G=\text{SO}(D+1)$ or $G=\text{SO}(1,D)$ canonical gauge theory over $\sigma$ with
 connection $A_{aIJ}$ and conjugate momentum $\pi^{aIJ}$. These variables are 
 subject to the canonical brackets
 \be \label{3.1}
 \{A_{aIJ}(x),\pi^{bKL}(y)\}=2 \beta \delta_a^b\delta_{[I}^K\delta_{J]}^L\delta^{(D)}(x-y),\;\;
 \{A_{aIJ}(x),A_{bKL}(y)\}=\{\pi^{aIJ}(x),\pi^{bKL}(y)\}=0\text{,}
 \ee
 as well as to the Gau{\ss} constraint
 \be \label{3.2}
 G^{IJ}:={\cal D}_a \pi^{aIJ}=\partial_a \pi^{aIJ}+2 A_a^{[I}\;_K\;\pi^{a|K|J]}
 \ee
 and the simplicity constraint
 \be \label{3.3}
 S^{ab}_{\overline{M}}=\frac{1}{4} \epsilon_{IJKL\overline{M}} \pi^{aIJ}\pi^{bKL} \text{.}
 \ee
 Internal indices as before are moved by the internal metric $\eta$ which is just 
 the Euclidean metric for SO$(D+1)$ ($\zeta=1$) and the Minkowski metric for SO$(1,D)$ ($\zeta=-1$).
 We have for $g\in \text{SO}(\zeta,D)$ that $g^{IJ} g^{KL}\eta_{KL}=\eta^{IJ},\;
 \det((g^{IJ}))=1$.
 The covariant differential ${\cal D}_a$ of $A$ acts only on internal indices. This 
 does not affect the tensorial character of (\ref{3.2}) because $\pi^{aIJ}$ is a Lie algebra 
 valued vector density of weight one and (\ref{3.2}) is its covariant divergence which is 
 independent of the Levi-Civita connection. The real parameter $\beta\not=0$ in (\ref{3.1}) 
 is similar to, but structurally different from the Immirzi parameter in $D=3$.\\
 \\
 Let $\Gamma_{aIJ}[\pi]$ be the hybrid connection (\ref{2.35}) constructed from $\pi$. 
 We define a map from this Yang-Mills theory phase space with coordinates 
 $(A_{aIJ},\pi^{aIJ})$ to the coordinates $(q_{ab},P^{ab})$ of the ADM phase space
 by the following formulas
 \ba \label{3.4}
 \det(q) q^{ab} &:=& \frac{1}{2\zeta} \pi^{aIJ} \; \pi^b\;_{IJ} \text{,}
 \\
 P^{ab} &:=& \frac{1}{4 \beta} \left(
 q^{a[c} [A_{cIJ}-\Gamma_{cIJ}]\pi^{b]IJ}
 +q^{b[c} [A_{cIJ}-\Gamma_{cIJ}]\pi^{a]IJ} \right) \nonumber \\
 &=& \frac{1}{2 \beta} q^{d(a} \;[A_{cIJ}-\Gamma_{cIJ}] \pi^{[b)IJ} \delta^{c]}_d \text{.}
 \nonumber
 \ea
 The central result of this section is:
 \begin{Theorem} \label{th3.1} ~\\
 i. Gau{\ss} and simplicity constraints obey a first class constraint algebra.\\
 ii. The symplectic reduction of the Yang-Mills phase space defined above with respect to 
 Gau{\ss} and simplicity constraints coincides with the ADM phase space. More in detail,
 the functions $q_{ab}[\pi],\;P^{ab}[A,\pi]$ defined in (\ref{3.4}) are Dirac observables with 
 respect to Gau{\ss} and simplicity constraints and obey the standard Poisson brackets
 \be \label{3.5}
 \{q_{ab}(x),P^{cd}(y)\}=\delta_{(a}^c \delta_{b)}^d\;\delta^{(D)}(x-y),\;\;
 \{q_{ab}(x),q_{cd}(y)\}=\{P^{ab}(x),P^{cd}(y)\}=0
 \ee
 on the constraint surface defined by simplicity and Gau{\ss} constraints.
 \end{Theorem}
 \begin{proof} ~\\
 i.\\ 
 Since $S^{ab}_{\overline{M}}$ only depends on $\pi^{aIJ}$, it Poisson commutes with itself.
 The Gau{\ss} constraint of course generates $G$ gauge transformations under which 
 $A$ transforms as a connection and $\pi$ as a section in an associated vector bundle 
 under the adjoint representation of $G$. The Poisson algebra of the smeared Gau{\ss} constraints
 is therefore (anti-)isomorphic with the Lie algebra of $G$
 \be \label{3.6}
 \{G[f],G[f']\}=-\beta G[[f,f']],\;G[f]:=\int\; d^Dx\; \frac{1}{2} f_{IJ} G^{IJ},\;\;[f,f']_{IJ}=2f_{[I}\;^K f'_{|K|J]} \text{.}
 \ee
 Under finite Gau{\ss} transformations we have 
 \be \label{3.7}
 \pi^{aIJ}\mapsto [g \pi^a g^{-1}]^{IJ} \text{.}
 \ee
 Since $G=\text{SO}(\zeta,D)$ is unimodular we obtain
 \be \label{3.8}
 S^{ab}_{\overline{M}}\mapsto \zeta \; g_{\overline{M}}\;^{\overline{N}} S^{ab}_{\overline{N}},\;\;\;\;\;
 g_{\overline{M}}\;^{\overline{N}}=\prod_{i=1}^{D-3} \; g_{M_i}\mbox{}^{N_i} \text{.}
 \ee
 It follows the first class structure $\{G,G\}\propto G,\;\; \{G,S\}\propto S,\;\; \{S,S\}=0$.\\
 ii.\\
The hybrid spin connection $\Gamma_{aIJ}[E]$ is a $G$ connection by construction. 
Its extension $\Gamma_{aIJ}[\pi]$ off the simplicity constraint surface therefore 
transforms as a $G$ connection modulo the simplicity constraint. 
 Since both $\pi^{aIJ},\;K_{aIJ}:=\frac{1}{\beta}(A_{aIJ}-\Gamma_{aIJ})$ transform in the adjoint representation
 of $G$ it is clear that $Q^{ab}\propto {\rm Tr}(\pi^a \pi^b),\;\;K_a^b\propto {\rm Tr}(K_a \pi^b)$
 are in fact Gau{\ss} invariant, possibly modulo the simplicity constraint, and thus are $q_{ab},\;P^{ab}$. Since $S^{ab}_{\overline{M}}$ and 
 $q_{ab}$ are both constructed from $\pi^{aIJ}$ alone it is clear that they strictly Poisson 
 commute. As for $P^{ab}$ we notice that it is a linear combination of the objects 
 \be \label{3.9}
 K_a\;^b:=- \frac{s}{4\beta}\;[A_{aIJ}-\Gamma_{aIJ}] \pi^{bIJ}\text{,}
 \ee
 with coefficients that depend only on $q_{ab}$. While the notation already suggests that $K_a\;^b$ is related with the extrinsic curvature, note that as it is defined here, $K_a\;^b$ has density weight one. It is therefore sufficient to show that
 $\{K_a\;^b,S^{cd}_{\overline{M}}\}\approx 0$. We compute with the smeared simplicity constraint
 and using that $\Gamma_{aIJ}[\pi]$ depends only on $\pi^{aIJ}$ 
 \ba \label{3.10}
 \{K_a\;^b(x),S[f]\} &=& 
 \int\; d^Dy\;
 f_{cd}^{\overline{M}}(y) \; \{K_a\;^b(x),S^{cd}_{\overline{M}}(y)\}
 \nonumber\\
 &=&
 - \frac{s}{16 \beta} \int\; d^Dy\;
 f_{cd}^{\overline{M}}(y)\; \pi^{bIJ}(x)\; \epsilon_{ABCD\overline{M}} \; 
 \{A_{aIJ}(x),\pi^{cAB}(y)\pi^{dCD}(y)\}
 \nonumber\\  
 &=& -s \; f_{cd}^{\overline{M}}(x) \delta_a^{(c}  S^{d) b}_{\overline{M}}(x)  \text{.}
 \ea
 It follows that $P^{ab}$ Poisson commutes with the simplicity constraint on its 
 constraint surface.\\
 \\
 It remains to verify the ADM Poisson brackets.\\ 
 Since $q_{ab}(x)$ depends only on $\pi^{aIJ}(x)$ we have trivially 
 $\{q_{ab}(x),q_{cd}(y)\}=0$. Next, using $q^{ab}=Q^{ab}/\det(q),\; \det(q)=[\det(Q)]^{1/(D-1)}$, we find
 \ba \label{3.11}
 \{q_{ab}(x),P^{cd}(y)\} 
 &=& 2s q_{ae}(x) \; q_{bf}(x) \;\{q^{ef}(x), (q^{g(c}\; K_h\mbox{}^{[d)}\delta^{h]}_g)(y)\}
 \nonumber\\
 &=& - \frac{1}{2\beta}  [q_{ae}\; q_{bf}](x) \; [q^{g(c} \pi^{[d)IJ} \delta^{h]}_g](y)
 \times \nonumber\\
 && \left[\frac{1}{\det(q)}
 \{Q^{ef}(x), A_{hIJ}(y)\}-\frac{1}{D-1} Q_{mn}(x) q^{ef}(x) \{Q^{mn}(x),A_{hIJ}(y)\} \right]
 \nonumber\\
 &=&- \frac{1}{2\beta} \delta^{(D)}(x-y)
 \; q_{ae}\; q_{bf} \; q^{g(c} \pi^{[d)IJ} \delta^{h]}_g \times \nonumber \\ 
 & & \frac{1}{2\zeta}
 \left[- \frac{4\beta}{\det(q)}\delta^{(e}_h \pi^{f)}_{IJ}
 +\frac{4\beta}{D-1} Q_{mn}(x) q^{ef} \delta^m_h \pi^n_{IJ}\right]  
 \nonumber\\
 &=& \delta^{(D)}(x-y)
 \; q_{ae}\; q_{bf} \; q^{g(c} 
 \left[q^{[d)|e|} \delta^{|f|}_h+q^{[d)|f|} \delta^{|e|}_h-\frac{2}{D-1}  q_{mn} q^{|ef|} q^{[d)|n|} 
 \delta^{|m|}_h \right]\; \delta^{h]}_g
 \nonumber\\
 &=& \delta^{(D)}(x-y) \delta^c_{(a} \delta^d_{b)}
 \ea
 by carefully contracting all indices and keeping track of the (anti)symmetrisations. \\
 The last bracket is the most complicated. We write 
 \be \label{3.12}
 P^{ab}=P^{abeIJ} K_{eIJ},\;\; P^{abeIJ}=\frac{1}{2}q^{g(a}\; \pi^{[b)IJ} \; \delta^{e]}_g
 \ee 
 and compute 
 \ba \label{3.13}
\{P^{ab}(x),P^{cd}(y)\} &=& 
P^{abeIJ}(x)\{K_{eIJ}(x),P^{cdfKL}(y)\} K_{fKL}(y)  \nonumber \\
& &-P^{cdfKL}(y)\{K_{fKL}(y),P^{abeIJ}(x)\} K_{eIJ}(x)
\nonumber\\
&& +P^{abeIJ}(x)\; P^{cdfKL}(y)\{K_{eIJ}(x),K_{fKL}(y)\}
\nonumber\\
&=& \frac{1}{2}\{P^{ab}(x),q^{h(c}(y)\} \pi^{[d)KL} \delta^{f]}_h \; K_{fKL}(y)
-\frac{1}{2}\{P^{cd}(y),q^{g(a}(x)\} \pi^{[b)IJ} \delta^{e]}_g\; K_{eIJ}(x)
\nonumber\\
&&
+\frac{1}{2\beta} P^{abeIJ}(x)\; q^{h(c}(y)\{A_{eIJ}(x),\pi^{[d)KL}(y)\} \delta^{f]}_h \; K_{fKL}(y)
 \nonumber\\
 && -\frac{1}{2\beta} P^{cdfKL}(y)\; q^{g(a}(x)\{A_{fKL}(y),\pi^{[b)IJ}(x)\} \delta^{e]}_g\; K_{eIJ}(x)
\nonumber\\
&&
 -\frac{1}{\beta^2} P^{abeIJ}(x) \; P^{cdfKL}(y)\;[
 \{A_{eIJ}(x),\Gamma_{fKL}(y)\}-\{A_{fKL}(y),\Gamma_{eIJ}(x)\}]
\nonumber\\
&=& 2s \left[q^{h(a} q^{b)(c }K_f^{[d)} \delta^{f]}_h 
-q^{g(c} q^{d)(a} K_e^{[b)} \delta^{e]}_g \right] \delta^{(D)}(x-y)
\nonumber\\
&&
-2s \left[\; q^{h(c} \delta_e^{[d)} \delta^{f]}_h \; q^{g(a} K_f^{[b)} \delta^{e]}_g
 -q^{g(a}\; \delta_f^{[b)} \delta^{e]}_g\;  q^{h(c} K_e^{[d)} \delta^{f]}_h \right]\delta^{(D)}(x-y)
\nonumber\\
&&
 -\frac{1}{\beta^2} P^{abeIJ}(x) \; P^{cdfKL}(y)\;[
 \{A_{eIJ}(x),\Gamma_{fKL}(y)\}-\{A_{fKL}(y),\Gamma_{eIJ}(x)\}] \text{.}
\ea
By carefully carrying out the contractions, it is not difficult to see that the first two square 
brackets in the last equality are each proportional to 
\be \label{3.14}
q^{ad} K^{[bc]}+q^{bd} K^{[ac]}+q^{ac} K^{[bd]}+q^{bc} K^{[ad]},\;\;K^{ab}:=q^{ac} K_c\;^b \text{.}
\ee
We claim that $K^{[ab]}$ is constrained to vanish by the Gau{\ss} constraint. To see this, 
let ${\cal D}'_a$ be the covariant differential of $A$ which acts also on tensor indices and 
let $D_a$ the covariant differential that kills the generalised $D$-bein. Then the Gau{\ss} 
constraint reads
\be \label{3.15}
G^{IJ}={\cal D}_a \pi^{aIJ}={\cal D}'_a \pi^{aIJ}=[{\cal D}'_a-D_a] \pi^{aIJ}+D_a \pi^{aIJ}
\approx [(A-\Gamma)_a,\pi^a]^{IJ}=\beta[K_a,\pi^a]^{IJ} \text{,}
\ee
where we used that on the constraint surface of the simplicity constraint we have 
$D_a \pi^{bKL}=2 D_a n^{[K} E^{a|L]} =0$. With the convention $K_{aI}:= - \zeta K_{aIJ} n^J$ we 
obtain for the Gau{\ss} constraint
\ba
G_{IJ} &=&2\beta K_{aL[I} \pi^a\;_{J]}\;^L \approx 2\beta
K_{aL[I} (n_{J]}E^{aL} - E^a_{J]} n^L ) \nonumber \\
 &=& -2 \zeta \beta K_{a[I} E^a_{J]} +2\beta K_{[I} n_{J]}=:\bar{G}_{IJ}+2  n_{[I} G_{J]}   \text{,} \label{3.16}
\ea
where $K_I=E^{aL} K_{aLI}$ is the trace part of $K_{aIJ}$. It follows that $\bar{K}_I=0$ and 
$K_{a[I} E^a_{J]}=0$ on the Gau{\ss} constraint surface. Now 
\ba 
K^{[ab]} E_{aI} E_{bJ} &\approx& -\frac{s\zeta}{2} q^{c[a} K_{cL} E^{b]L} E_{aI} E_{bJ}
= -\frac{s\zeta}{2} q^{ca} K_{cL} E^{bL} E_{a[I}  E_{bJ]} \nonumber \\ &=& -\frac{s\zeta}{2 \det{(q)}} K_{cL} \bar{\eta}^L_{[J} E^c_{I]}
= \frac{s\zeta}{2\det{(q)}} K_{a[I} E^a_{J]} \text{.} \label{3.17}
\ea
Therefore $K^{[ab]}=[K^{[cd]} E_{cI} E_{cJ}] E^{aI} E^{bJ}$ vanishes on the Gau{\ss} constraint surface. 

We now turn to the last square bracket in (\ref{3.13}). It is a linear combination, with coefficients
$M$ depending on $q_{ab}$, of expressions of the form
\be \label{3.18}
M^{abe}_f(x) M^{cdg}_h(y)
\pi^{fIJ}(x) \pi^{hKL}(y) \left[\{A_{eIJ}(x),\Gamma_{gKL}(y)\}-\{A_{gKL}(y),\Gamma_{eIJ}(x)\} \right] \text{.}
\ee
We now invoke the key result of the previous section and write 
$\Gamma_{aIJ}=\delta F/\delta\pi^{aIJ}+S_{aIJ}$ where $S_{aIJ}$ vanishes on the 
constraint surface of the simplicity constraint and depends at most on its first partial 
derivatives. It is therefore given by an expression of the form
\be \label{3.19}
S_{gKL}=\lambda_{gKLgmn}^{\overline{M}} S^{mn}_{\overline{M}}+\mu_{gKLmn}^{\overline{M}p}
\partial_p S^{mn}_{\overline{M}}
\ee
for certain coefficients $\lambda,\mu$.
First of all, we notice that due to the commutativity of partial functional derivatives
\be \label{3.20}
\{A_{eIJ}(x),\delta F/\delta\pi^{gKL}(y)\}-\{A_{gKL}(y),\delta F/\delta\pi^{eIJ}(x)\}=0 \text{.}
\ee
Next, due to the derivatives involved, the Poisson bracket is not ultralocal, however, what we intend 
to prove is that $\{P[f],P[f']\}\approx 0$ with the smeared functions $P[f]=\int d^Dx f_{ab} P^{ab}$.
Let $M^e_f=f_{ab} M^{abe}_f,\;M^{\prime g}_h=M^{cdg}_h f'_{cd}$, then the contribution 
from $S_{gKL}$ in the first term of (\ref{3.18}) becomes after smearing
\begin{alignat}{3} \label{3.21}
\approx&\; 
\int\; d^Dx\;\int\; d^Dy\; M^e_f(x) \pi^{fIJ}(x) 
\left( M^{\prime g}_h \pi^{hKL} \lambda_{gKLmn}^{\overline{M}}-[M^{\prime g}_h \pi^{hKL}
\mu_{gKLmn}^{\overline{M}p}]_{,p}\right)\; \{A_{eIJ}(x),S^{mn}_{\overline{M}}(y)\}
\nonumber\\
=&\; 4\beta
\int\; d^Dx\; M^e_f \pi^{fIJ}(x) 
\left( M^{\prime g}_h \pi^{hKL} \lambda_{gKLmn}^{\overline{M}}-[M^{\prime g}_h \pi^{hKL}
\mu_{gKLmn}^{\overline{M}p}]_{,p}\right)\; 
\delta_e^{(m} S^{n)f}_{\overline{M}}(x)
\nonumber\\
\approx &\; 0 \text{,}
\end{alignat}
where (\ref{3.10}) was used.
The calculation for the second term is similar. In conclusion, $\{P^{ab}(x),P^{cd}(y)\}$
vanishes on the joint constraint surface of the Gau{\ss} and the simplicity constraint. 
\end{proof}

\section{ADM Constraints in Terms of the New Variables}
\label{s4}

It remains to express the ADM constraints in terms of the new variables. Of course we could 
just substitute for the expressions (\ref{3.4}), however, this is not the most convenient form 
for the ADM constraints because they involve the hybrid connection which is a complicated 
expression in terms of $\pi$. We will therefore adopt the strategy familiar from $D+1=4$ and 
invoke the curvature $F$ of $A$. In the end, we will arrive at expressions $\mathcal{H}_a,\mathcal{H}$ for spatial diffeomorphism and Hamiltonian constraint which differ 
from their counterparts $\mathcal{H}'_a,\mathcal{H}'$, obtained by naive substitution of $q_{ab},P^{ab}$ 
by (\ref{3.4}) in (\ref{2.2}), (\ref{2.3}), by terms proportional to Gau{\ss} 
and simplicity constraints. This guarantees that the algebra of Gau{\ss}, simplicity, spatial 
diffeomorphism an Hamiltonian constraints is of first class:\\ 
To see this, let us write 
$\mathcal{H}_a=\mathcal{H}'_a+Z_a,\;\mathcal{H}=\mathcal{H}'+Z$ where $Z_a,Z$ vanish 
on the constraint surface of the simplicity and Gau{\ss} constraint. 
We have seen already that $\{S,S\}=0,\{G,S\}\propto S,\;\{G,G\}\propto G$. We also have shown 
that (\ref{3.4}) are weak Dirac observables with respect to $S$ and invariant under $G$. Since $\mathcal{H}'_a,\mathcal{H}'$ are defined in 
terms of (\ref{3.4}) it follows that $\{S,\mathcal{H}'_a\}\propto S, \{S,\mathcal{H}'\}\propto S$. Altogether therefore 
$\{S,\mathcal{H}_a\}, \; \{S,\mathcal{H}\},\; \{G,\mathcal{H}_a\},\;\{G,\mathcal{H}\} \propto S,G$ thus $S,G$ form an ideal.
Next we have $\{\mathcal{H}'_a,\mathcal{H}'_b\}\propto \mathcal{H}'_c,S,G,\;\{\mathcal{H}'_a,\mathcal{H}'\}\propto \mathcal{H}',S,G,\; \{\mathcal{H}',\mathcal{H}'\}\propto
\mathcal{H}'_a,S,G$ because the algebra of the variables (\ref{3.4}) is the same as that of the ADM variables 
modulo $S,G$ terms and therefore the algebra of the ADM constraints is reproduced modulo 
$S,G$ terms. Together with what was already said, this implies that $\mathcal{H}_a,\mathcal{H}$ reproduce the 
ADM algebra of constraints modulo $S,G$ terms.\\
\\
We begin by deriving the relation between the hybrid curvature
\be \label{4.1}
R_{abIJ}=2\partial_{[a}\Gamma_{b]IJ}+\Gamma_{aIK}\Gamma_b\;^K\;_J-\Gamma_{aJK}
\Gamma_b\;^K\;_I
\ee
and the  Riemann curvature of $q_{ab}$. Let $\nabla_a$ be the covariant derivative compatible with $q_{ab}$. Then we have by definition $D_a e_b^I=\nabla_a e_b^I+\Gamma_a\,^I\,_J e_b^J=0$.
Expanding out $D=\nabla+\Gamma$ in the commutator relation $[D_a,D_b] e_c^I=0$ and using
$[\nabla_a,\nabla_b]e_c^I=R_{abc}\;^d e_d^I$ we find
\be \label{4.2}
R_{abc}\;^d e_d^I+ R_{ab}\;^I\;_J e_c^J=0\;\;\Rightarrow\;\; R_{abcd}=R_{abIJ} e^I_c e^J_d \text{.}
\ee
This relation looks familiar from the spin connection, but we stress $\Gamma_{aIJ}$ is 
not the spin connection because $e_a^I$ is not a $D$-bein.
We  obtain modulo $S$ for the Ricci scalar
\be \label{4.3}
R_{abIJ} \pi^{aIK} \pi^b\;_K\;^J\approx R_{abIJ}[ n^I E^{aK}  -  n^K E^{aI}][ n_K E^{bJ} - n^J E^b_K]
=-\zeta \det(q) R \text{.}
\ee
Next, using (\ref{4.2}) 
\be \label{4.4}
R_{abIJ} \pi^{bIJ}\approx 2 R_{abIJ}  n^I E^{bJ}= 2 q^{bc}\sqrt{\det(q)} R_{abIJ} n^I e_c^J
= - 2 q^{bc} \sqrt{\det(q)} R_{abc}\;^d e_{dI} n^I=0 \text{,}
\ee
which is the analog of the algebraic Bianchi identity. \\
\\ 
We now expand the curvature 
\be \label{4.5}
F_{abIJ}:=2\partial_{[a}A_{b]IJ}+A_{aIK} \;A_b\;^K\;_J-A_{aJK}\; A_b\;^K\;_I
\ee
of $A=\Gamma+\beta K$ in terms of $\Gamma, K$ and obtain
\be \label{4.6}
F_{abIJ}=R_{abIJ}+2\beta D_{[a} K_{b]IJ}+2\beta^2 K_{[aIK}\; K_{b]}\;^K\;_{J} \text{,}
\ee
where torsion freeness of $\nabla=D-\Gamma$ was employed. Contracting (\ref{4.6})
with $\pi^{bIJ}$ we find using (\ref{4.4})
\be \label{4.7}
F_{abIJ} \pi^{bIJ}\approx  2\beta (D_{[a} K_{b]IJ}) \pi^{bIJ} - \beta^2 {\rm Tr}([K_a,K_b] \pi^b) \text{.}
\ee
The second term is proportional to the Gau{\ss} constraint because 
${\rm Tr}([K_a,K_b] \pi^b)={\rm Tr}(K_a[K_b,\pi^b])$ and  remembering (\ref{3.15}). 
In the first term we notice that $D_a \pi^{bIJ}\approx 0$ so that 
\be \label{4.8}
F_{abIJ} \pi^{bIJ}\approx -8s\beta D_{[a} K_{b]}^b= 4s\beta D_b[K_a\;^b-\delta_a^b K_c\;^c]
=-4\beta D_b P_a\;^b= 2\beta \mathcal{H}_a
\ee
is proportional to the spatial diffeomorphism constraint modulo $S,G$. \\
\\
Next, using (\ref{4.3})
\be \label{4.9}
F_{abIJ} \pi^{aIK}\pi^b\;_K\;^J
\approx -\zeta \det(q) R+2\beta D_a{\rm Tr}(K_b [\pi^a,\pi^b])-\beta^2 {\rm Tr}([K_a,K_b] \pi^a \pi^b) \text{.}
\ee
The second term is again proportional to the Gau{\ss} constraint, since 
${\rm Tr}(K_b[\pi^a,\pi^b])=-{\rm Tr}(\pi^a[K_b,\pi^b])$. So far all the steps were similar to the 
$3+1$ situation. The difference comes in when looking at the third term in (\ref{4.9})
\ba \label{4.10}
-{\rm Tr}([K_a,K_b]\pi^a \pi^b) &\approx&
[K_{aIK} K_b\,^K\,_J-K_{bIK} K_a\,^K\,_J][n^IE^{aL} -  n^L E^{aI} ][ n^L E^{bJ}- n^J E^b_L ]
\nonumber\\
&=& -\zeta [-(K_{aIK} E^{aI})(K_{bJ}\;^K E^{bJ})+(K_{bIK} E^{aI})(K_{aJ}\;^K E^{bJ})] \text{.}
\ea
By the Gau{\ss} constraint (\ref{3.16}) we have $K_I=K_{aJI} E^{aJ}=\zeta [K_J n^J] n_I$ and 
$K_J n^J\approx - K_{aIJ} \pi^{aIJ}/2=2s K_a^a$. Thus the first term in (\ref{4.10}) is given by
$4 [K_a^a]^2$. However, the second term cannot be written in terms of $K_a^b$.
To explore the structure of the disturbing term we notice that from $\bar{K}_I=0$ we have the 
decomposition
\be \label{4.11}
K_{aIJ}=\bar{K}^T_{aIJ}+2  n_{[I}K_{a|J]},\;\;K_{aI}=-\zeta K_{aIJ} n^J \text{.}
\ee
Hence 
\ba \label{4.12}  
&& -\zeta (K_{bIK} E^{aI})(K_{aJ}\;^K E^{bJ})=
-\zeta (\bar{K}^T_{bIK} E^{aI} - K_{bI} E^{aI} n_K)(\bar{K}^T_{aJ}\;^K E^{bJ} - K_{aJ} E^{bJ} n^K)
\nonumber\\
&=& -\zeta (\bar{K}^T_{bIK} E^{aI})(\bar{K}^T_{aJ}\;^K E^{bJ})- 4 K_a^b K_b^a \text{,}
\ea
where $K_{aI} E^{bI}= - \zeta K_{aIJ} E^{bI} n^J\approx K_{aIJ} \pi^{bIJ}/(2\zeta)$ was used.
Altogether
\be \label{4.13}
-{\rm Tr}([K_a,K_b]\pi^a \pi^b)=-4[K_a^b K_b^a-(K_c^c)^2]-\zeta (\bar{K}^T_{bIK} E^{aI})
(\bar{K}^T_{aJ}\;^K E^{bJ}) \text{.}
\ee
The first term in (\ref{4.13}) has the structure that appears in the Hamiltonian constraint and 
can be written in terms of $P^{ab},q_{ab}$, however, the second term does not appear in the 
Hamiltonian constraint and must be removed. Also notice that the Ricci term has sign $-\zeta$ while
the first term has negative sign. If we are interested in Lorentzian Gravity then the relative sign
between these two terms should be negative which is not the case for the choice of 
a compact gauge group $\zeta=1$. Therefore the expression (\ref{4.9}) fails to yield the 
Hamiltonian constraint for several reasons.\\
\\
To assemble the Hamiltonian constraint without making use of $\Gamma$, the idea is to consider
covariant derivatives which give access to $A$. Using suitable algebraic combinations 
then yields the desired expressions. To that end, let again ${\cal D}_a$ be the covariant 
differential of $A$ acting only on internal indices and let ${\cal D}'_a$ be its extension by 
the Levi-Civita connection. Consider  
\be \label{4.14}
D_b\;^a:=\pi^{aK}\m_J \; ({\cal D}_b \pi^{cJL})\; \pi_{cKL}
=\pi^{aK}\m_J \; ({\cal D}'_b \pi^{cJL})\; \pi_{cKL}
-2 \pi^{aK}\m_J\; \pi_{cKL} \;\Gamma^{[c}_{bd} \pi^{d]JL} \text{.}
\ee
The second term equals modulo $S$ 
\be \label{4.15a}
-2[n^K E^a_J-n_J E^{aK} ][n_K  E_{cL} - n_L E_{cK} ] \Gamma^{[c}_{bd} \pi^{d]JL}
=-2\zeta\; E^a_J \;E_{cL} \Gamma^{[c}_{bd} \pi^{d]JL} \approx 0
\ee
and thus vanishes modulo $S$. Writing ${\cal D}'_a=[{\cal D}'_a-D_a]+D_a$ and noticing
$D_a \pi^{cJL}\approx 0$ we obtain
\ba \label{4.15b}
\pi^{aK}\;_J \; ({\cal D}_b \pi^{cJL})\; \pi_{cKL}
&\approx& \zeta \beta  E^a_J E_{cL} [K_b\;^J\;_M \pi^{cML}+K_b\;^L\;_M\;\pi^{cJM}]
\nonumber\\
&\approx& \zeta \beta E^a_J E_{cL} [K_b\;^J\;_M E^{cL} n^M-K_b\;^L\;_M E^{cJ} n^M]
\nonumber\\
&=& -\beta (D-1) E^{aJ}\; K_{bJ}=  2s\zeta\beta(D-1) K_b\;^a \text{.}
\ea
It follows that 
\be \label{4.16}
\frac{1}{(D-1)^2}[D_b\;^a \; D_a\;^b-(D_c\;^c)^2]\approx  4\beta^2 [K_b\;^a\;K_a\;^b-
(K_c^c)^2]
\ee
and thus linear combinations of (\ref{4.9}) and (\ref{4.16}) can be used in order to produce 
the correct factor in front of the term quadratic in the extrinsic curvature.\\
\\
In analogy to (\ref{4.14}), consider
\be \label{4.17}
D^{aIJ}:=\pi^{b[I}\;_K {\cal D}_b \pi^{a|K|J]}=
\pi^{b[I}\;_K {\cal D}'_b \pi^{a|K|J]}-2
\pi^{b[I}\;_K \Gamma^{[a}_{bc} \pi^{c]|K|J]} \text{.}
\ee
The second term equals modulo $S$
\be \label{4.18}
2\zeta E^{b[I} \Gamma^{[a}_{bc} E^{c]J]}=(-\zeta \Gamma^c_{bc} E^{b[I}) E^{aJ]}
\ee
and thus is pure trace. Since we intend to cancel $\bar{K}^T_{aIJ}$ we therefore consider 
instead of (\ref{4.17}) its transverse tracefree projection
\be \label{4.19}
\bar{D}_T^{aIJ}:=[P_{TT}\cdot D]^{aIJ},\;\;[P_{TT}]^{aIJ}_{bKL}=\delta^a_b \bar{\eta}^I_{[K} 
\bar{\eta}^J_{L]}-\frac{2}{D-1}E^{a[I} \bar{\eta}^{J]}_{[K} E_{bL]}\text{,}
\ee
under which (\ref{4.18}) drops out.
The projector $P_{TT}$ can be expressed purely in terms of $\pi^{aIJ}$ using  (\ref{2.31}) and 
\be \label{4.20}
E^{a[I} \bar{\eta}^{J]}_{[K} E_{bL]}\approx -\zeta \left( \pi^{aM[I} \bar{\eta}^{J]}_{[K} \pi_{bL]M}
+\delta^a_b n^{[I} \delta^{J]}_{[K} n_{L]} \right) \text{.}
\ee
We continue using again $D_a \pi\approx 0$
\ba \label{4.21}
\bar{D}_T^{aIJ} &\approx& 
\beta P_{TT}\left( \pi^{b[I|K|} \left[K_{bKL} \pi^{a|L|J]}+K_b\,^{J]}\,_L \pi^a\;_K\;^L \right] \right)
\nonumber\\
&\approx& 
-\beta \zeta\;
P_{TT}\left(E^{b[I} K_b\,^{J]}\,_L E^{aL}\right)=-\beta \zeta \; E^{b[I}\bar{K}^{J]L}_{bT} E^a_L \text{.}
\ea
Notice that the last line is indeed tracefree and transverse. 
We write (\ref{4.21}) as 
\be \label{4.22}
\bar{D}_T^{aIJ}=\frac{\beta}{4}  F^{aIJ,bKL} \;\bar{K}^T_{bKL},\;\;\;\; F^{aIJ,bKL}=4 \zeta E^{b[I} \bar{\eta}^{J][L} E^{a K]} \text{.}
\ee
The tensor $F^{aIJ,bKL}$ can be seen as bilinear form on transverse tensors of type 
$\bar{K}_{aIJ}$ and has the following inverse
\be \label{4.23}
(F^{-1})_{aIJ,bKL}= \frac{\zeta}{4}[Q_{ab}\bar{\eta}_{[K|[I} \bar{\eta}_{J]|L]}-2E_{b[I} \bar{\eta}_{J][K} E_{aL]}] \text{,}
\ee
that is $[F\cdot F^{-1}]^{aIJ}_{bKL}=\delta^a_b \bar{\eta}^I_{[K}\bar{\eta}^J_{L]}$. Using 
(\ref{2.31}) and 
\be \label{4.24}
E_{aI} E_{bJ}\approx \zeta [\pi_{aIM} \pi_{bJ} \m^M-\zeta Q_{ab} n_I n_J] \text{,}
\ee
$F^{-1}$ is completely expressed in terms of $\pi^{aIJ}$. The quadratic combination of 
$\bar{K}^T$ to be removed from (\ref{4.13}) can now be compactly written as
\ba \label{4.25}
&&E^{bI}\bar{K}^T_{bJM} E^{aJ}\bar{K}^T_{aI}\;^M=E^{b[I} \bar{\eta}^{N][M} E^{aJ]}
\bar{K}^T_{bJM} \bar{K}^T_{aIN}
\nonumber\\
&=& \frac{\zeta}{4} F^{aIN,bJM} \bar{K}^T_{aIN} \bar{K}^T_{bJM}
= 4 \frac{\zeta}{\beta^2}(F^{-1})_{aIJ,bKL} \bar{D}_T^{aIJ} \bar{D}_T^{bKL} \text{.}
\ea
We now have all the pieces we need. The appropriate Hamiltonian constraint for 
spacetime signature $s$ is displayed in (\ref{2.3}). We find 
\begin{alignat}{3} \label{4.26}
\sqrt{\det(q)} \mathcal{H} =&
\zeta \left( F_{abIJ} \pi^{aIK} \pi^b\,_K\,^J+ 4 \bar{D}_T^{aIJ}\;(F^{-1})_{aIJ,bKL}\; \bar{D}_T^{bKL}+\frac{1}{(D-1)^2}[D_b\,^a \, D_a\,^b-(D_c\,^c)^2]\right)  \nonumber\\
&  - s\frac{1}{\beta^2} \frac{1}{(D-1)^2}[D_b\,^a \, D_a\,^b-(D_c\,^c)^2] \text{.}
\end{alignat}
This expression simplifies for $s=\zeta$ and $\beta=1$ in which case the terms quadratic in 
$D_b\;^a$ precisely cancel. This is again similar to the situation in $3+1$ dimensions.
This special situation can also be obtained more directly starting from the Palatini formulation
as we will see in \cite{BTTII}.

\section{Conclusion}
\label{s5}

In this paper, we succeeded in constructing a Hamiltonian connection formulation of General
Relativity in all spacetime dimensions $D+1\ge 3$ based on the gauge group 
SO$(D+1)$ or SO$(1,D)$. In addition to the usual Gau{\ss}, spatial
diffeomorphism and Hamiltonian constraints, there are simplicity constraints that dictate
that the momentum conjugate to the connection is determined by a generalised $D$-bein.
The theory can be constructed for all four possible combinations of the internal ($\zeta$)
and spacetime ($s$) signature. This is especially attractive with an eye towards quantisation
because unique \cite{LewandowskiUniquenessOfDiffeomorphism} background independent representations of spatially diffeomorphism invariant
theories of connections with compact structure group exist in any dimension and have been studied in great detail
(see, e.g., \cite{ThiemannModernCanonicalQuantum} and references therein).  

The techniques for quantising Gau{\ss}, spatial
diffeomorphism and Hamiltonian constraint that have been developed in $3+1$ dimensions
generalise to arbitrary dimensions as we will show in \cite{BTTIII}. The 
simplicity constraint provides a challenge. A similar kind of constraint plays a prominent role in Spin Foam models 
\cite{BarrettRelativisticSpinNetworks, EngleFlippedSpinfoamVertex, FreidelANewSpin} and various proposals for its quantisation have been made. The problem 
is that the quantum simplicity constraints tend to be anomalous. This is due to the fact that 
the classically commuting $\pi^{aIJ}$ become non commuting after discretization (Spin foams) or introduction of a singular smearing (canonical approach), a property which is then shared by the corresponding operators in the quantum theory.
In \cite{BTTV} we propose some new strategies for how to make progress on this issue.
Eventually, the solution of the simplicity constraint will consist in a restriction on the
set of labels for spin network functions. 

The application of interest of the present work is of course in higher dimensional Supergravity
theories. Here we have to face two new technical challenges: For Lorentzian Supergravity 
the action is formulated in terms of a Lorentzian internal metric which would naturally
imply the choice SO$(1,D)$. Hence, in order to keep SO$(D+1)$ we must carefully disassemble
the SO$(1,D)$ Clifford algebra and reassemble it into an SO$(D+1)$ Clifford algebra which 
turns out to be possible. The second challenge is that higher dimensional Supergravity theories depend next to the Rarita-Schwinger field also on higher $p$-form fields for 
which background independent Hilbert space representations first need to be developed. 

While many interesting technical issues are not settled by our analysis, the present work 
and its continuation in the companion papers hopefully contribute to the development
of a non perturbative definition of quantum (Super)gravity in any dimension.   \\
\\
\\
\\
{\bf\large Acknowledgements}\\
NB and AT thank the German National Merit Foundation for financial support. 
The part of the research performed at the Perimeter Institute for
Theoretical Physics was supported in part by funds from the Government of
Canada through NSERC and from the Province of Ontario through MEDT.

\newpage

\begin{appendix}

\section{Independent Set of Simplicity Constraints}
\label{fixedpoint}

The result that one would like to prove is as follows:
\begin{Theorem} \label{th1} ~\\
Let $\pi^{aIJ}$ be a tensor antisymmetric in $I,J$ and $a=1,..,D;\;
I,J=1,..,D+1$ subject to the condition that for any non zero vector the $D$ vectors $\pi^a;\;\pi^{aI}:= -\zeta \pi^{aI}\,_J n^J$ are linearly independent.  
Then it is possible to construct a tensor 
$E^a_I=E^a_I[\pi]$ with the following properties:\\
Let $n^I=n^I[\pi]$ be the unique normal satisfying 
$E^a_I n^I=0,\;n^I n^J\eta_{IJ}=\zeta$ where $\zeta$ corresponds to the signature
of $\eta$. Let $\bar{\pi}^{aIJ}=\bar{\eta}^I_K\;\bar{\eta}^J_L \pi^{aKL}$
with the transversal projectors $\bar{\eta}^I_J=\delta^I_J-\zeta n^I\; n_J$.
Then $\bar{\pi}^{aIJ}=\bar{\pi}^{aIJ}_T$ is automatically tracefree 
with respect to $E$, that is, $\bar{\pi}^{aIJ} E_{aI}=0$ where $E_{aI}$ 
is uniquely defined by $E^{aI} E_{aJ}=\bar{\eta}^I_J,\;E^{aI} E_{bI}=
\delta^a_b$. Furthermore $\pi^{aIJ}=\bar{\pi}^{aIJ}+2 n^{[I}E^{a|J]} $.
\end{Theorem}
In what follows we describe some ideas towards a possible  proof.\\ 
\\
Given $\pi^{aIJ}$, let $n^I$ be {\it any} unit vector to begin with and 
construct
$\bar{\eta}_{IJ}$ as above. Define $E^{aI}[\pi,n]:=- \zeta \pi^{aI}\,_J n^J$. Notice 
that automatically $E^{aI} n_I=0$. Then we obtain the decomposition
\be \label{a}
\pi^{aIJ}=\bar{\pi}^{aIJ}+2 n^{[I}E^{a|J]} \text{.}
\ee
It is interesting to note that the non zero vector (due to the assumed linear
independence)  
\be \label{b}
N_I[\pi,n]:=\epsilon_{I J_1..J_D}\epsilon_{a_1..a_D}\;
E^{a_1 J_1}[\pi,n]..E^{a_D J_D}[\pi,n]
\ee
coincides up to normalisation with $n_I$ no matter what $\pi$ is. Furthermore 
\be \label{b2}
\eta^{IJ} N_I N_J=\zeta \; [D!]^2\; \det(Q);\;Q^{ab}:=\eta^{IJ} E^a_I E^b_J \text{,}
\ee
where $\zeta =\pm 1$ if $\eta$ has Euclidean or Lorentzian signature respectively.
In particular, we verify that for $\zeta =-1$, the vector $N_I$ is timelike, null
or spacelike if and only if $n_I$ is.  

The tensor 
$E_{aI}[\pi,n]$ is given, up to normalisation, by
\be \label{c}
E_{aI}\propto -\epsilon_{IJ_1.. J_D}\; \epsilon_{a a_2 .. a_D}\; n^{J_1}\;
E^{a_2 J_2}\;..\;E^{a_D J_D} \text{.}
\ee
The condition that $\bar{\pi}^{aIJ}$ be tracefree with respect to $E$
becomes 
\be \label{d}
\bar{\pi}^{aIJ} E_{aI}=[\pi^{aIJ}-2n^{[I}E^{a|J]}]\; E_{aI}
=\pi^{aIJ}\;E_{aI}+ \; D\; n^J=0 \text{.}
\ee
We can reformulate this as the condition that $\pi^{aIJ} E_{aI}$ is 
longitudinal, the coefficient of proportionality then follows from 
the normalisations. We define the tensor 
\be \label{e}
\kappa_{IJ_1..J_D}[\pi]:=\epsilon_{I K_1..K_D}\; \epsilon_{a_1..a_D}
\pi^{a_1 K_1}\;_{J_1}\;..\;\pi^{a_D K_D}\;_{J_D} \text{,}
\ee
which is totally symmetric in $J_1,..,J_D$ and only depends on $\pi$, not
on $n$. In terms of this tensor the tracefree condition becomes 
\be \label{e1}
\pi^{aIJ} E_{aI}\propto \kappa_I\;^J\;_{J_2 .. J_D}\; n^I\; 
n^{J_2} .. n^{J_D} \stackrel{!}{\propto} n^J \text{.}
\ee
Using the normalisation condition we can write this as the equality
\be \label{f}
\kappa_{J_1 I J_3..J_{D+1}}\; n^{J_1}\; n^{J_3}\;..\; n^{J_{D+1}}=\zeta n_I\; 
\kappa_{J_1..J_{D+1}} n^{J_1}\;..\;n^{J_{D+1}} \text{.}
\ee
This is a system of $D$ independent non-polynomial equations of order $D+2$
for the $D$ independent unknowns $n^I,\;I=1,..,D;\;n^{D+1}=
\pm \sqrt{1-\zeta \sum_{I=1}^D (n^I)^2}$. We can turn this into an equivalent 
system of $D+1$ homogeneous, polynomial equations of order $D+2$
by    
\be \label{g}
\kappa_{J_1 I J_3..J_{D+1}}\; n_{J_2}\; n^{J_1}\;..\; n^{J_{D+1}}= n_I\; 
\kappa_{J_1..J_{D+1}} n^{J_1}\;..\;n^{J_{D+1}} \text{,}
\ee
which leaves the normalisation of $n^I$ undetermined. 

Up to this point, $n$ was just an extra structure independent of and 
next to $\pi$.
The idea is now to fix $n^I$ in terms of $\pi$ by solving the system 
(\ref{f}). Having 
determined $n^I=n^I[\pi]$ would then yield the desired tensor 
$E^{aI}[\pi]:=E^{aI}[\pi,n[\pi]]$. However, it is far 
from clear whether a solution exists, nor that it is unique, 
although the number
of independent equations matches with the number of degrees of freedom to 
be fixed. Being polynomial, it is clear that {\it complex} solutions of 
(\ref{g}) exist, but the system of equations is far too complex in order 
to see whether {\it real} solutions exist. Hence to secure at least existence,
we must resort to different methods.

Since the polynomial formulation (\ref{g}) is of no help, we stick with 
(\ref{f}). We write it in the form
\be \label{h}
n^I=\zeta \frac{f^I(n)}{n_J f^J(n)},\;\;f_I(n)=\kappa_{J_1}\;^I\;_{J_3..J_{D+1}}
n^{J_1}\; n^{J_3}\;..\; n^{J_{D+1}} \text{.}
\ee
This equation takes the form of a {\it fixed point equation} $x=f(x)$. 
In order to apply established theorems, (\ref{h}) is not useful because 
fixed point theorems typically are for compact sets and the right hand side 
of (\ref{h}) has not manifestly bounded range, especially for signature
$\zeta=-1$. 

Consider instead the function
\be \label{i}
F^I(n):=||n||\;[2-||n||]\frac{f^I(n)}{||f(n)||} \text{.}
\ee
Notice that we use the {\it Euclidean} metric in both numerator and 
denominator also for the 
case $\zeta=-1$, i.e. $||n||^2:=\delta_{IJ} n^I n^J$.  
Let us now restrict $n^I$ to the compact (closed and bounded) and 
convex\footnote{Suppose that $||u||,||v||\le 1$ then  
$$
||s u+(1-s)v||^2=s^2 ||u||^2+(1-s)^2 ||v||^2+2s(1-s)<u,v>
\le s^2 ||u||^2+(1-s)^2 ||v||^2+2s(1-s)||u|| \; ||v||
\le 1
$$ 
for any $s\in [0,1]$ due to the Cauchy Schwarz inequality.} 
$(D+1)$-ball 
\be \label{j}
B_{D+1}:=\{n\in \mathbb{R}^{D+1};\;\delta_{IJ} n^I n^J\le 1\} \text{.}
\ee
The map (\ref{i}) is a continuous map from 
$B_{D+1}$ to itself. To see this, notice that $||f(n)||\not=0$ except at 
$||n||=0$. This follows from the identity
\be \label{j_1}
f_I(n) n^I=N_I(n) n^I \text{,}
\ee
which for $\zeta=1$ and for $\zeta=-1$ and $n$ not null shows that $||f||\not=0$
unless $||n||=0$.
For $\zeta=-1$ and $n$ null we have in fact $f_I=\gamma n_I$ with 
$\gamma\not=0$. To see this, notice that the span of the $E^{aI},\;a=1,..,D$
contains $n^I$. Introduce some basis of the orthogonal complement, say
$b_\alpha^I,\; \alpha=2,.., D$ and let $b_1^I=n^I$ so that 
$\eta_{IJ} b^I_\alpha b^J_\beta=\delta_{\alpha\beta},\;\alpha,\beta=2,..,D$ and 
$b_1^I b_{\alpha I}=0$. Then we have an expansion $E^{aI}=r^{a\alpha} 
b^I_\alpha$ with $\det(r)\not=0$ due to linear independence by assumption. 
Next, there exists $\gamma\not=0$ such that 
\be \label{j_2}
\epsilon_{IJ_1..J_D} 
b^{J_1}_{\alpha_1} .. b^{J_D}_{\alpha_D}=\gamma n_I 
\epsilon_{\alpha_1..\alpha_D} \text{,}
\ee
since the $b_\alpha$ are linearly independent and the left hand side of 
(\ref{j_2}) is orthogonal to all of them, hence it must be null. We can 
therefore compute
\begin{eqnarray} \label{j_3}
f_I
&=& \delta_1^{\alpha_1} b_{\alpha_1}^{J_1}\epsilon_{J_1 K_1 .. K_D}
\epsilon_{a_1..a_D}\pi^{a_1 K_1}\;_I E^{a_2 K_2} .. E^{a_D K_D}
\nonumber\\
&=& -\delta_1^{\alpha_1} [\epsilon_{J K_1 .. K_D}
b^{K_1}_{\alpha_1}.. b^{K_D}_{\alpha_D}]
\epsilon_{a_1..a_D}\pi^{a_1 J}\;_I r^{a_2 \alpha_2} .. r^{a_D \alpha_D}
\nonumber\\
&=& -\gamma \epsilon_{1\alpha_2..\alpha_D} n_J
\epsilon_{a_1..a_D}\pi^{a_1 J}\;_I r^{a_2 \alpha_2} .. r^{a_D \alpha_D}
\nonumber\\
&=& \gamma \epsilon_{1\alpha_2..\alpha_D} 
\epsilon_{a_1..a_D}E^{a_1}_I r^{a_2 \alpha_2} .. r^{a_D \alpha_D}
\nonumber\\
&=& \gamma \epsilon_{1\alpha_2..\alpha_D} b_{\beta I}
\epsilon_{a_1..a_D}r^{a_1\beta} r^{a_2 \alpha_2} .. r^{a_D \alpha_D}
\nonumber\\
&=& \gamma \epsilon_{1\alpha_2..\alpha_D} b_{\beta I}
\epsilon_{\beta\alpha_2..\alpha_{D-1}}\det(r)
\nonumber\\
&=& \gamma [(D-1)!] \det(r) n_I \text{.}
\end{eqnarray}

It follows that (\ref{i}) is everywhere well defined except possibly at 
$||n||=0$ where the fraction $f(n)/||f(n)||$ is ill defined. However, 
due to the prefactor $||n||$ we see that $F(n):=0$ at $||n||=0$ is a 
continuous extension. Next  
\be \label{j_4}
||F(n)||=||n||(2-||n||)=1-[1-||n||]^2\in [0,1]\text{,}
\ee
hence $F$ maps $B_{D+1}$ to itself. By the {\it Brouwer Fixed Point Theorem} 
\cite{ReedBook1} applicable to compact convex subsets of Euclidean space,
it has a fixed point, that is, the equation $n^I=F^I(n)$ has at least 
one solution, a fixed point $n^I=n^I_\ast[\pi]$. Unfortunately, this is not 
very helpful because $n=0$ is a trivial fixed point and the Brouwer fixed 
point theorem does not tell us anything about the number of fixed points,
hence it could be that $n=0$ is the only one. However, notice that 
$||F(n)||\ge ||n||$ and $||F(n)||=||n||\Leftrightarrow ||n||=1$. This 
suggests that if a fixed point can be found by iteration $n_{k+1}:=F(n_k)$
then it will lie on the sphere $S^D$. Indeed for $||n||=1$ the map $F$
maps $S^D$ to itself. 

In order to see whether a fixed point can be obtained by using iteration 
methods we estimate for $n_1,n_2\in S^D$
\begin{eqnarray} \label{j_5}     
||F(n_1)-F(n_2)||^2 &=& 
2[1-\frac{<f_1,f_2>}{||f_1||\;||f_2||}]]
\nonumber\\
&=& \frac{1}{||f_1||\; ||f_2||}[||f_1-f_2||^2-[||f_1||-||f_2||]^2]
\nonumber\\
&\le & 
\frac{||f_1-f_2||^2}{||f_1||\; ||f_2||} \text{.}
\ea
Now
\be \label{j_6}
D_{IJ}:=\partial f_I/\partial n_J=[\kappa_{JIJ_2..J_D}+(D-1)\kappa_{IJ J_2..J_D}]
n^{J_2}..n^{J_D} \text{.}
\ee
It follows $D_{IJ} n^J=D f_I$.
For $n_2$ sufficiently close to $n_1$, we obtain with the Cauchy Schwarz 
inequality
\be \label{j_7}
||f_2-f_1||^2\approx ||D(n_1)[n_2-n_1]||^2\le {\rm Tr}(D^T(n_1) D(n_1)) 
\; ||n_2-n_1||^2
\ee
and for $||n||=1$ again due to the Cauchy Schwarz inequality
\be \label{j_8}
{\rm Tr}(D^T D)=\sum_I\{[\sum_J [D_{IJ}]^2] ||n||^2\}\ge 
\sum_I [\sum_J D_{IJ} n^J]^2\ge D^2 ||f||^2 \text{.}
\ee
Thus the right hand side of (\ref{j_5}) is given by $q(n_1,n_2) ||n_2-n_1||^2$,
where $q(n,n)\ge D^2$. Hence $F$ fails to be a contraction map and we cannot
invoke techniques familiar from the {\it Banach Fixed Point Theorem} 
\cite{ReedBook1} in order to prove existence of a fixed point as this would need 
$\sup_{n_1,n_2} q(n_1,n_2)<1$. Either sharper bounds are needed or we have to
use a different iteration function (recall that fixed point equations can
be written in many different but equivalent ways and for some of them 
the iteration map maybe contractible, for others not). Notice that as long as 
we are only interested in obtaining $f(n)\propto n$ we may rescale $f$ by 
a sufficiently large constant such that $f$ itself becomes a contraction map
and maps $B_{D+1}$ to itself. This is possible because $f$ and the 
matrix defined by $f(n_2)-f(n_1)=D(n_2,n_1)\cdot(n_2-n_1)$ are continuous 
maps on the compact sets $B_{D+1}$ and $B_{D+1}\times B_{D+1}$ respectively and 
thus are uniformly bounded. However, due to (\ref{j_8}) the required constant
would turn $f$ into a strictly norm decreasing map and thus can only have 
$n=0$ as a fixed point.\\
\\
\\
Remark:\\
In contrast to $D$ odd, for $D$ even there are two natural vectors that
one construct purely from $\pi$ namely 
\be \label{s}
u_I=\kappa_{I J_1.. J_D} \eta^{J_1 J_2}..\eta^{J_{D-1} J_D},\;\;
v_I=\kappa_{J_1 I J_2 .. J_D} \eta^{J_1 J_2}..\eta^{J_{D-1} J_D} \text{.}
\ee
These are the only independent contractions that exist because 
$\kappa_{I J_1 .. J_D}$ is completely symmetric in its $J$ indices. 
It is natural to assume that the fixed point vector is a linear 
combination of $u,v$ and indeed in $D=2$ it is easy to see that $n\propto u$.
For $D\ge 4$ we were not able to verify this by direct calculation or any
other means due to the complexity of the fixed point equation.

\end{appendix}

\newpage

\bibliography{pa87pub.bbl}

\providecommand{\href}[2]{#2}\begingroup\raggedright\begin{thebibliography}{10}

\bibitem{ArnowittTheDynamicsOf}
R.~Arnowitt, S.~Deser, and C.~W. Misner, ``{The dynamics of general
  relativity},'' in {\em Gravitation: An introduction to current research}
  (L.~Witten, ed.), (New York), pp.~227--265, Wiley, 1962.
\newblock {\tt arXiv:gr-qc/0405109}.

\bibitem{GoroffQuantumGravityAt}
M.~H. Goroff and A.~Sagnotti, ``{Quantum gravity at two loops},'' {\em Physics
  Letters B} {\bf 160} (1985) 81--86.

\bibitem{GoroffTheUltravioletBehavior}
M.~H. Goroff and A.~Sagnotti, ``{The ultraviolet behavior of Einstein
  gravity},'' {\em Nuclear Physics B} {\bf 266} (1986) 709--736.

\bibitem{DeserTwoOutcomesFor}
S.~Deser, ``{Two outcomes for two old (super)problems},'' in {\em The many
  faces of the Superworld, Yuri Goldfand memorial volume} (M.~Shifman, ed.),
  World Publishing2000.
\newblock {\tt arXiv:hep-th/9906178}.

\bibitem{DeserInfinitiesInQuantum}
S.~Deser, ``{Infinities in quantum gravities},'' {\em Annalen der Physik} {\bf
  9} (2000) 299--306, {\tt arXiv:gr-qc/9911073}.

\bibitem{DeserNonrenormalizabilityOfLast}
S.~Deser, ``{Nonrenormalizability of (last hope) D= 11 supergravity, with a
  terse survey of divergences in quantum gravities},'' {\tt
  arXiv:hep-th/9905017}.

\bibitem{GreenBook1}
M.~B. Green, J.~H. Schwarz, and E.~Witten, {\em {Superstring Theory, Vol. 1:
  Introduction}}.
\newblock Cambridge University Press, Cambridge, 1988.

\bibitem{GreenBook2}
M.~B. Green, J.~H. Schwarz, and E.~Witten, {\em {Superstring Theory, Vol. 2:
  Loop Amplitudes, Anomalies and Phenomenology}}.
\newblock Cambridge University Press, Cambridge, 1988.

\bibitem{PolchinskiBook1}
J.~Polchinski, {\em {String Theory, Vol. 1: An Introduction to the bosonic
  string}}.
\newblock Cambridge University Press, Cambridge, 1998.

\bibitem{PolchinskiBook2}
J.~Polchinski, {\em {String Theory, Vol. 2: Superstring theory and beyond}}.
\newblock Cambridge University Press, Cambridge, 1998.

\bibitem{ThiemannModernCanonicalQuantum}
T.~Thiemann, {\em {Modern Canonical Quantum General Relativity}}.
\newblock Cambridge University Press, Cambridge, 2007.

\bibitem{AshtekarNewVariablesFor}
A.~Ashtekar, ``{New Variables for Classical and Quantum Gravity},'' {\em
  Physical Review Letters} {\bf 57} (1986) 2244--2247.

\bibitem{BarberoRealAshtekarVariables}
J.~Barbero, ``{Real Ashtekar variables for Lorentzian signature space-times},''
  {\em Physical Review D} {\bf 51} (1995) 5507--5510, {\tt
  arXiv:gr-qc/9410014}.

\bibitem{ThiemannQSD1}
T.~Thiemann, ``{Quantum spin dynamics (QSD)},'' {\em Classical and Quantum
  Gravity} {\bf 15} (1998) 839--873, {\tt arXiv:gr-qc/9606089}.

\bibitem{ImmirziRealAndComplex}
G.~Immirzi, ``{Real and complex connections for canonical gravity},'' {\em
  Classical and Quantum Gravity} {\bf 14} (1997) L177--L181, {\tt
  arXiv:gr-qc/9612030}.

\bibitem{AshtekarRepresentationsOfThe}
A.~Ashtekar and C.~J. Isham, ``{Representations of the holonomy algebras of
  gravity and non-Abelian gauge theories},'' {\em Classical and Quantum
  Gravity} {\bf 9} (1992) 1433--1468, {\tt arXiv:hep-th/9202053}.

\bibitem{AshtekarRepresentationTheoryOf}
A.~Ashtekar and J.~Lewandowski, ``{Representation Theory of Analytic Holonomy
  C* Algebras},'' in {\em Knots and Quantum Gravity} (J.~Baez, ed.), (Oxford),
  Oxford University Press1994.
\newblock {\tt arXiv:gr-qc/9311010}.

\bibitem{LewandowskiUniquenessOfDiffeomorphism}
J.~Lewandowski, A.~Okol\'{o}w, H.~Sahlmann, and T.~Thiemann, ``{Uniqueness of
  Diffeomorphism Invariant States on Holonomy-Flux Algebras},'' {\em
  Communications in Mathematical Physics} {\bf 267} (2006) 703--733, {\tt
  arXiv:gr-qc/0504147}.

\bibitem{FleischhackRepresentationsOfThe}
C.~Fleischhack, ``{Representations of the Weyl Algebra in Quantum Geometry},''
  {\em Communications in Mathematical Physics} {\bf 285} (2009) 67--140, {\tt
  arXiv:math-ph/0407006}.

\bibitem{NietoTowardsAnAshtekar8}
J.~A. Nieto, ``{Towards an Ashtekar formalism in eight dimensions},'' {\em
  Classical and Quantum Gravity} {\bf 22} (2005) 947--955, {\tt
  arXiv:hep-th/0410260}.

\bibitem{NietoTowardsAnAshtekar12}
J.~A. Nieto, ``{Towards an Ashtekar formalism in 12 dimensions},'' {\em General
  Relativity and Gravitation} {\bf 39} (2007) 1109--1119, {\tt
  arXiv:hep-th/0506253}.

\bibitem{NietoOrientedMatroidTheory}
J.~A. Nieto, ``{Oriented matroid theory and loop quantum gravity in (2+2) and
  eight dimensions},'' {\em Revista mexicana de fisica} {\bf 57} (2011)
  400--405, {\tt arXiv:1003.4750 [hep-th]}.

\bibitem{MeloschNewCanonicalVariables}
S.~Melosch and H.~Nicolai, ``{New canonical variables for d=11 supergravity},''
  {\em Physics Letters B} {\bf 416} (1998) 91--100, {\tt arXiv:hep-th/9709227}.

\bibitem{HanHamiltonianAnalysisOf}
M.~Han, Y.~Ma, Y.~Ding, and L.~Qin, ``{Hamiltonian analysis of n-dimensional
  Palatini gravity with matter},'' {\em Modern Physics Letters} {\bf A20}
  (2005) 725--732, {\tt arXiv:gr-qc/0503024}.

\bibitem{AshtekarNewHamiltonianFormulation}
A.~Ashtekar, ``{New Hamiltonian formulation of general relativity},'' {\em
  Physical Review D} {\bf 36} (1987) 1587--1602.

\bibitem{PeldanActionsForGravity}
P.~Peldan, ``{Actions for gravity, with generalizations: A Review},'' {\em
  Classical and Quantum Gravity} {\bf 11} (1994) 1087--1132, {\tt
  arXiv:gr-qc/9305011}.

\bibitem{CarlipQuantumGravityIn}
S.~Carlip, {\em {Quantum Gravity in 2+1 Dimensions}}.
\newblock Cambridge University Press, Cambridge, 2003.

\bibitem{BTTII}
N.~Bodendorfer, T.~Thiemann, and A.~Thurn, ``{New variables for classical and
  quantum gravity in all dimensions: II. Lagrangian analysis},'' {\em Classical
  and Quantum Gravity} {\bf 30} (2013) 045002, {\tt arXiv:1105.3704 [gr-qc]}.

\bibitem{BTTIII}
N.~Bodendorfer, T.~Thiemann, and A.~Thurn, ``{New variables for classical and
  quantum gravity in all dimensions: III. Quantum theory},'' {\em Classical and
  Quantum Gravity} {\bf 30} (2013) 045003, {\tt arXiv:1105.3705 [gr-qc]}.

\bibitem{BTTV}
N.~Bodendorfer, T.~Thiemann, and A.~Thurn, ``{On the implementation of the
  canonical quantum simplicity constraint},'' {\em Classical and Quantum
  Gravity} {\bf 30} (2013) 045005, {\tt arXiv:1105.3708 [gr-qc]}.

\bibitem{BTTIV}
N.~Bodendorfer, T.~Thiemann, and A.~Thurn, ``{New variables for classical and
  quantum gravity in all dimensions: IV. Matter coupling},'' {\em Classical and
  Quantum Gravity} {\bf 30} (2013) 045004, {\tt arXiv:1105.3706 [gr-qc]}.

\bibitem{BTTVI}
N.~Bodendorfer, T.~Thiemann, and A.~Thurn, ``{Towards loop quantum supergravity
  (LQSG): I. Rarita-Schwinger sector},'' {\em Classical and Quantum Gravity}
  {\bf 30} (2013) 045006, {\tt arXiv:1105.3709 [gr-qc]}.

\bibitem{BTTVII}
N.~Bodendorfer, T.~Thiemann, and A.~Thurn, ``{Towards loop quantum supergravity
  (LQSG): II. p -form sector},'' {\em Classical and Quantum Gravity} {\bf 30}
  (2013) 045007, {\tt arXiv:1105.3710 [gr-qc]}.

\bibitem{AlexandrovSU(2)LoopQuantum}
S.~Alexandrov and E.~Livine, ``{SU(2) loop quantum gravity seen from covariant
  theory},'' {\em Physical Review D} {\bf 67} (2003) 044009, {\tt
  arXiv:gr-qc/0209105}.

\bibitem{MitraGaugeInvariantReformulationAnomalous}
P.~Mitra and R.~Rajaraman, ``{Gauge-invariant reformulation of an anomalous
  gauge theory},'' {\em Physics Letters B} {\bf 225} (1989) 267--271.

\bibitem{MitraGaugeInvariantReformulation}
P.~Mitra and R.~Rajaraman, ``{Gauge-invariant reformulation of theories with
  second-class constraints},'' {\em Annals of Physics} {\bf 203} (1990)
  157--172.

\bibitem{AnishettyGaugeInvarianceIn}
R.~Anishetty and A.~S. Vytheeswaran, ``{Gauge invariance in second-class
  constrained systems},'' {\em Journal of Physics A: Mathematical and General}
  {\bf 26} (1993) 5613--5619.

\bibitem{VytheeswaranGaugeUnfixingIn}
A.~S. Vytheeswaran, ``{Gauge unfixing in second-class constrained systems},''
  {\em Annals of Physics} {\bf 236} (1994) 297--324.

\bibitem{FreidelBFDescriptionOf}
L.~Freidel, K.~Krasnov, and R.~Puzio, ``{BF description of higher-dimensional
  gravity theories},'' {\em Advances in Theoretical and Mathematical Physics}
  {\bf 3} (1999) 1289--1324, {\tt arXiv:hep-th/9901069}.

\bibitem{BarrettRelativisticSpinNetworks}
J.~W. Barrett and L.~Crane, ``{Relativistic spin networks and quantum
  gravity},'' {\em Journal of Mathematical Physics} {\bf 39} (1998) 3296--3302,
  {\tt arXiv:gr-qc/9709028}.

\bibitem{EngleFlippedSpinfoamVertex}
J.~Engle, R.~Pereira, and C.~Rovelli, ``{Flipped spinfoam vertex and loop
  gravity},'' {\em Nuclear Physics B} {\bf 798} (2008) 251--290, {\tt
  arXiv:0708.1236 [gr-qc]}.

\bibitem{FreidelANewSpin}
L.~Freidel and K.~Krasnov, ``{A new spin foam model for 4D gravity},'' {\em
  Classical and Quantum Gravity} {\bf 25} (2008) 125018, {\tt arXiv:0708.1595
  [gr-qc]}.

\bibitem{ReedBook1}
M.~Reed and B.~Simon, {\em {Methods of Modern Mathematical Physics, Vol. 1:
  Functional Analysis}}.
\newblock Academic Press, 1981.

\end{thebibliography}\endgroup

\end{document}